\documentclass[lettersize,journal]{IEEEtran}
\usepackage{amsmath,amsfonts,algorithmic,array,textcomp,stfloats}
\usepackage{verbatim}
\usepackage{cite}
\usepackage{amssymb,ifpdf}
\usepackage[pdftex]{graphicx}
\graphicspath{{../pdf/}{../jpeg/}{../pic/}{../../allpic/}{./pic/}{./work4_3_pic/}{./Authors/}}
\DeclareGraphicsExtensions{.pdf,.jpeg,.png}

\usepackage[para]{footmisc}
\usepackage[font=normalsize,labelfont=sf,textfont=sf]{caption}
\usepackage[font=normalsize,labelfont=sf,textfont=sf]{subcaption}

\usepackage{amssymb}  
\usepackage{csquotes} 
\usepackage{lmodern}  
\usepackage{mathrsfs} 
\usepackage[ruled,linesnumbered]{algorithm2e} 
\usepackage{mathtools} 
\usepackage{cases} 
\usepackage{empheq} 
\usepackage[capitalize]{cleveref} 
\usepackage{cuted} 
\usepackage{bm} 
\usepackage{amsthm} 
\usepackage{tabulary} 
\usepackage[flushleft]{threeparttable}
\usepackage{tensor}
\usepackage{multirow}
\usepackage{tabularx}
\usepackage{booktabs}
\usepackage[table]{xcolor}
\usepackage{optidef} 

\usepackage[hyphens]{url}
\expandafter\def\expandafter\UrlBreaks\expandafter{\UrlBreaks
  \do\a\do\b\do\c\do\d\do\e\do\f\do\g\do\h\do\i\do\j%
  \do\k\do\l\do\m\do\n\do\o\do\p\do\q\do\r\do\s\do\t%
  \do\u\do\v\do\w\do\x\do\y\do\z\do\A\do\B\do\C\do\D%
  \do\E\do\F\do\G\do\H\do\I\do\J\do\K\do\L\do\M\do\N%
  \do\O\do\P\do\Q\do\R\do\S\do\T\do\U\do\V\do\W\do\X%
  \do\Y\do\Z}

\newtheorem{theorem}{Theorem}

\theoremstyle{definition}

\newcommand{\bb}[1]{\mathbb{#1}}

\newcommand{\mcal}[1]{\mathcal{#1}}

\newcommand{\tf}{\tilde{f}}
\newcommand*\dif{\mathop{}\!\mathrm{d}} 

\newcommand{\rom}[1]{(\lowercase\expandafter{\romannumeral #1\relax})}

\newcommand*\circlenum[1]{\raisebox{.5pt}{\textcircled{\raisebox{-.9pt} {#1}}}}

\begin{document}
\title{A Dynamic Hierarchical Framework for IoT-assisted Metaverse Synchronization}
\author{Yue~Han,
Dusit~Niyato,~\IEEEmembership{Fellow,~IEEE},
Cyril~Leung, ~\IEEEmembership{Life Member,~IEEE},
Dong~In~Kim,~\IEEEmembership{Fellow,~IEEE},
Kun~Zhu,~\IEEEmembership{Member,~IEEE},
Shaohan~Feng,
Xuemin(Sherman)~Shen,~\IEEEmembership{Fellow,~IEEE},
Chunyan~Miao,~\IEEEmembership{Senior Member,~IEEE}
\thanks{
Y. Han is with Alibaba Group and the Alibaba-NTU Joint Research Institute (JRI), Nanyang Technological University (NTU), Singapore. E-mail: hany0028@e.ntu.edu.sg.}
\thanks{D. Niyato and C. Miao are with the School of Computer Science and Engineering (SCSE), NTU, Singapore. E-mail: \{dniyato,ascymiao\}@ntu.edg.sg.}
\thanks{C. Leung is with The University of British Columbia, Canada. E-mail: cleung@ece.ubc.ca.}
\thanks{D.I. Kim is with the Department of Electrical and Computer Engineering, Sungkyunkwan University, Suwon, South Korea. E-mail: dikim@skku.ac.kr.}
\thanks{K. Zhu is with the College of Computer Science and Technology, Nanjing University of Aeronautics and Astronautics, Nanjing 210016, China. E-mail: zhukun@nuaa.edu.cn}
\thanks{S. Feng is with Institute for Infocomm Research (I2R), A*STAR, Singapore. E-mail:Feng\_Shaohan@i2r.a-star.edu.sg}
\thanks{X. Shen is with the Department of Electrical and Computer Engineering, University of Waterloo, Waterloo, ON N2L 3G1, Canada. E-mail:sshen@uwaterloo.ca.}
}

\markboth{Journal of \LaTeX\ Class Files,~Vol.~14, No.~8, August~2021}%
{Shell \MakeLowercase{\textit{et al.}}: A Sample Article Using IEEEtran.cls for IEEE Journals}


\maketitle

\begin{abstract}
Metaverse has recently attracted much attention from both academia and industry. Virtual services, ranging from virtual driver training to online route optimization for smart goods delivery, are emerging in the Metaverse. To make the human experience of virtual life more real, digital twins (DTs), namely digital replications of physical objects, are key enablers. However, the status of a DT may not always accurately reflect that of its real-world twin because the latter may be subject to changes with time. As such, it is necessary to synchronize a DT with its physical counterpart to ensure that its status is accurate for virtual businesses. In this paper, we propose a dynamic hierarchical framework in which a group of Internet-of-Things (IoT) devices are incentivized to sense and collect physical objects' status information collectively so as to assists virtual service providers (VSPs) in synchronizing DTs. Based on the collected synchronization data and the value decay rate of the DTs, the VSPs can determine synchronization intensities to maximize their payoffs. 
In our proposed dynamic hierarchical framework, the lower-level evolutionary game captures the VSPs selection by the population of IoT devices, and the upper-level differential game captures the VSPs payoffs, which are affected by the synchronization strategy, IoT devices selections, and the DTs value status, given VSPs are simultaneous decision makers. We further consider the case in which some VSPs are first movers (i.e., there is sequential decision making among VSPs), and extend it as a Stackelberg differential game.  We theoretically and experimentally show that the equilibrium to the lower-level game exists and is evolutionarily robust, and provide a sensitivity analysis with respect to various system parameters. Experiments show that the both the simultaneous and Stackelberg differential game give higher accumulated payoffs compared to the baseline of a static game. 
\end{abstract}

\begin{IEEEkeywords}
Metaverse, IoTs, game theory, digital twins, resource allocation, crowdsensing
\end{IEEEkeywords}

\section{Introduction}
\IEEEPARstart{T}{he} long-term global epidemic has dramatically altered people's work and life styles. In-person social gatherings and events have gradually been replaced by various online events, e.g., UC Berkeley held its virtual commencement in $2021$ \cite{WatchBlockeleyUC2020} and \enquote{Fortnite} organized a virtual concert in $2019$, reportedly viewed by $10.7$ million people \cite{FortniteMarshmelloConcert}. These virtual events are examples of virtual business operations, supported by virtual business providers (VSPs). Such terms can be encapsulated by the concept of the Metaverse, which has been attracting interest from both academia and industry since $2019$. 

Metaverse is known as a platform to provide social, immersive, and interactive experiences with perpetual user accounts \cite{schroederSocialInteractionVirtual2002}. The key objective of the Metaverse is to host different digital/virtual worlds, in which different VSPs provide their virtual business, such as retailing \cite{gadallaMetaverseretailServiceQuality2013}, gaming \cite{volkCocreativeGameDevelopment2008}, education \cite{diazVirtualWorldResource2020},  and social-networking \cite{schroederSocialInteractionVirtual2002}, and people as the Metaverse users (hereinafter referred to as \textit{users}), can fully immerse themselves in the virtual life and experience various virtual services via their avatars, namely digital replicas of themselves. The Metaverse opens up a new world for both VSPs and users. 
On the one hand, VSPs benefit by providing virtual services in the Metaverse. For example, even with its borders closed due to a pandemic, a destination country can still generate revenue from its tourism industry by providing\textit{ virtual traveling services} \cite{kwokCOVID19ExtendedReality2021}. A global manufacturer can build a \textit{digital factory} in the Metaverse, to replicate the operational processes of multiple sites in the real world, so that different optimization strategies can be quickly tested with minimal interference to the real production \cite{leeCyberPhysicalSystemsArchitecture2015}. On the other hand, it benefits the users. For example, frequent travelers can overcome the physical challenges of in-person travel to a destination country that has closed its borders due to a pandemic, and still enjoy a \textit{virtual travel} in the Metaverse \cite{michelVirtualTravelExperiences}. A procurement manager (buyer) can virtually check on the production of goods she has ordered via the \textit{virtual factory} and efficiently negotiate terms and conditions.
It is expected that the benefits, conveniences and employment opportunities provided by the Metaverse will transform people's lives in a similar way that the Internet has done. This explains why the Metaverse is sometimes referred to as Internet of $3$D virtual worlds\cite{dionisio3DVirtualWorlds2013}.

Notwithstanding its great potential benefits, the Metaverse is still in the nascent stages of its development\cite{duanMetaverseSocialGood2021}. One challenge is how to efficiently replicate the real world in the Metaverse. One solution is through synchronization of digital twins (DTs) \cite{khanDigitalTwinEnabled6GVision2021}, which are digital replicas of living or non-living entities in the real world, with bi-directional communication between each DT and its physical entity. With DTs, Metaverse users can experience physical entities as if they were interacting with them in the real world, and VSPs can develop businesses based on DTs. 
An example is virtual driver training \cite{taheriVirtualRealityDriving2017}, in which users are trained on a simulation platform with immersive realistic experience and minimal danger. 
Here, roads, cars, traffic, passengers, and even weather can be replicated in the Metaverse. Based on the Metaverse's interoperability \cite{radoffWeb3InteroperabilityMetaverse2021},  DTs can be reused for other VSPs to support e.g., telemedicine \cite{bhattTargetedApplicationsUnmanned2018} and online route optimization for smart goods delivery \cite{sanjabProspectTheoryEnhanced2017}.
\begin{figure*}[]
\begin{minipage}[b]{\linewidth}
\centering
\includegraphics[width=0.95\linewidth]{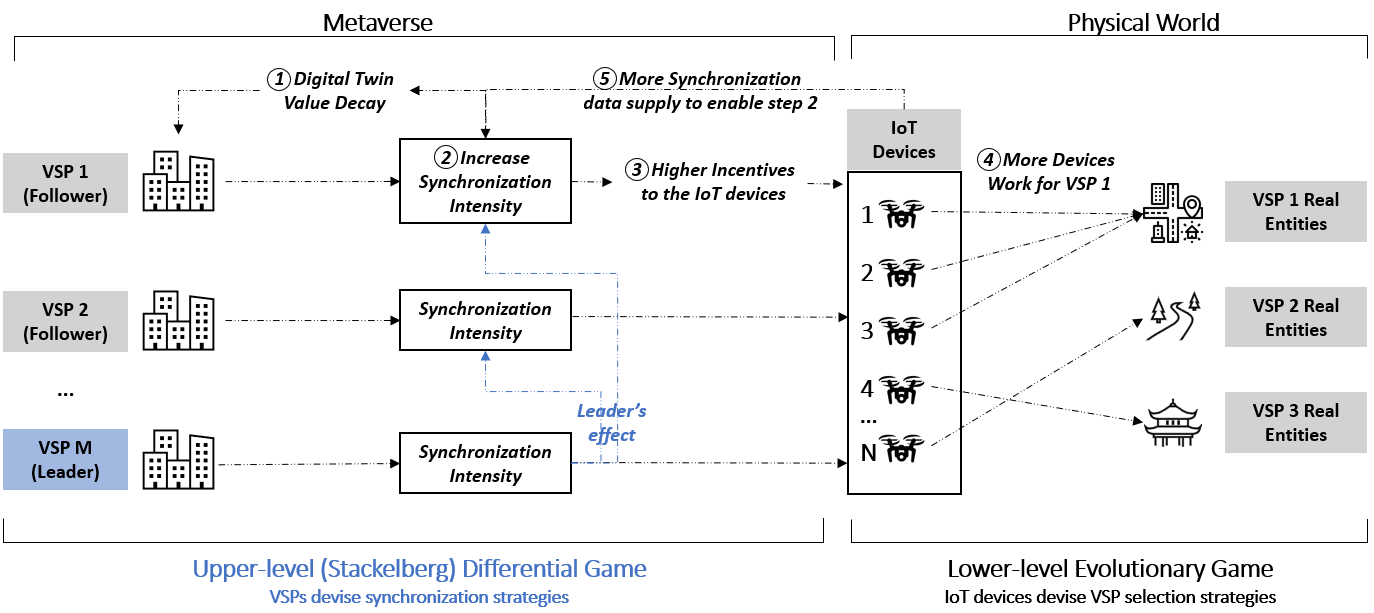}
\caption{A dynamic hierarchical framework for IoT-assisted Metaverse synchronization}
\label{fig-system-model}
\end{minipage} 
\end{figure*}

However, there are four challenges. First, there are various types of VSPs in the Metaverse and VSPs can have ad-hoc requests for additional DTs to augment virtual businesses. Therefore, it can be expensive and challenging to deploy traditional fixed and static sensors to collect the status data of physical objects. This raises the question of how to leverage \textit{movable} IoT devices to support the Metaverse synchronization. 
Second, the synchronization between DTs and their physical counterparts is needed to keep the DT states up to date and maintain DT values (characterized by e.g. reliability \cite{khanDigitalTwinEnabled6GVision2021}) for VSPs' virtual businesses. If there is a lack of proper synchronization between the virtual and real worlds, unreliable DTs can affect VSPs' business profits. Hence, an efficient resource allocation solution is needed, according to which a number of movable IoT devices can be effectively allocated to sense the corresponding physical counterparts of DTs for each of the VSPs. 
Third, it is possible that different VSPs have different levels of tolerance to  non-updated DTs, given their different business types, the extents to which they use AI, as well as computation capabilities. For example, an non-updated weather twin may leave little impact on the virtual travel service provider, but a significant impact on virtual driver training by making the simulated environment less realistic. A VSP with a stronger AI algorithm deployed and higher computation capacity can predict the patterns of DT states in the near future, and thus may have higher tolerance to non-updated DTs. The different tolerance to non-updated DTs allow VSPs optimize their synchronization intensities differently. Therefore, how to represent, model, and take into account the VSP's tolerance levels in the process of determining the optimal synchronization intensity is critical and remains as an open research problem. 
Fourth, similar to real-world economics markets, there may be some influential decision makers in the Metaverse market, such as some VSPs with large business volumes. As such, how their decisions about optimal synchronization intensity affect the other following VSPs is not addressed in the literature.


To address the above challenges, in this paper, we propose a dynamic hierarchical Metaverse synchronization framework utilizing Internet of Things (IoT) devices, in which \textit{movable} IoT devices such as drones (UAVs) can collectively sense the current states of the real entities that are linked to the DTs of the VSP. 
The system model is shown in \cref{fig-system-model} and consists of Metaverse components and Physical World components. In the Metaverse, we consider DT values to be relevant to the business profitability of VSPs and subject to natural value decay, i.e., the value naturally drops with time (Step \circlenum{1}). As such, it is necessary for VSPs to determine a proper synchronization intensity to maintain DT values without incurring excessive costs due to over-synchronization (Step \circlenum{2}). To enable this, the VSPs provide incentives to attract IoT devices to work for them by sensing the corresponding real world entities (Step \circlenum{3}). However, faced with different VSPs and their different incentive provisions, an IoT device can independently select a VSP that maximizes its utility and work for it (Step \circlenum{4}). After completing the sensing task for the selected VSP, IoT devices transmit the data to the VSP through nearby base stations (Step \circlenum{5}). With proper cloud processing, VSPs can use synchronization data to update their DTs accordingly, thereby increasing the values of the DTs. In analogous to physical world economic markets, we consider that larger VSPs in the Metaverse may have certain levels of privilege, that is, to determine their synchronization strategies in advance of other smaller VSPs. We refer to a large VSP that moves first as a leader (highlighted in blue in \cref{fig-system-model}) and smaller VSPs that first observe the leader's synchronization strategies as followers.

We model the above system model as a hierarchical two-level game. The lower-level game captures UAVs' VSP selection strategies (Step \circlenum{4}) as an evolutionary game and the upper-level game models VSP synchronization strategy (Step \circlenum{2}) as a simultaneous differential game (when no leader exists) and a Stackelberg differential game (in the leader-follower case). The optimal synchronization strategy is solved by the open-loop solution grounded in optimal control theory. 
To summarize, the main contributions of the paper are as follows:
\begin{itemize}
\item 
We propose a general IoT-assisted synchronization data collection framework by leveraging movable IoT devices to collect physical objects state data. The flexibility of movable devices enables VSPs to have ad-hoc applications and faster virtual business expansion, especially critical for the budding phase of the Metaverse. 
\item Due to the self-interested nature, each device can independently select a VSP to work for. We adopt an evolutionary game to capture their VSP selection behaviors based on bounded rationality. 
\item We address the optimal synchronization strategies for VSPs by formulating it as a simultaneous differential game and solve it based on optimal control theory. In particular, we propose DT value dynamics that capture a VSP's level of tolerance to a non-updated DT. The DT value dynamics as well as the IoT devices' VSP selection dynamics are jointly considered in the system states for the optimal control problem. 
\item  We finally extend the dynamic hierarchical framework by considering leaders and followers among VSPs and formulate it as a differential Stackelberg game. Extensive experiments show that both simultaneous and Stackelberg differential game lead to higher accumulated utility for each VSP, compared to the baseline of a static game.
\end{itemize}

The rest of this paper is organized as follows. \cref{sec-related-work}
reviews related works. \cref{sec-system-model} presents the system
model and hierarchical dynamic game framework. UAVs' selection of VSPs is formulated as an evolutionary game in \cref{sec-lower-evoluitaonry-game}. In Sections \ref{sec-upper-diff-game} and \ref{sec-upper-hierachical-play}, we formulate the problem as simultaneous differential game and Stackelberg differential game respectively to solve for the optimal synchronization strategy. \cref{sec-experiments} presents numerical results and sensitivity analysis. \cref{sec-conclusion} concludes the paper.

\section{Related Work}\label{sec-related-work}
\subsection{Metaverse and its Architecture}
Recently, due to the COVID-19 pandemic restrictions and marketing by leading technology companies such as Facebook and Microsoft \cite{brownBigTechWants2021}, the convenience and usefulness of the virtual world (Metaverse) have been highlighted. Despite some early discussion of the benefits and challenges brought by virtual services in terms of user experiences \cite{gadallaMetaverseretailServiceQuality2013,volkCocreativeGameDevelopment2008,diazVirtualWorldResource2020,schroederSocialInteractionVirtual2002}, novel applications in the Metaverse have attract research attention, e.g., the authors in \cite{duanMetaverseSocialGood2021} build a university campus prototype to study social good in the Metaverse, whereas the authors in \cite{nguyenMetaChainNovelBlockchainbased2021} propose a Blockchain-based framework for Metaverse applications. An incentive mechanism design for leveraging coded distributed computing for Metaverse services is studied in \cite{jiangReliableCodedDistributed2021}. In \cite{limRealizingMetaverseEdge2022}, the interconnection of edge intelligence and the Metaverse is considered. These examples illustrate the increasing importance of realizing the Metaverse. However none of them consider how to improve the convergence between the Metaverse and the physical world, in particular, the synchronization intensity problem for VSPs in terms of their DTs.  The authors in \cite{hanDynamicResourceAllocation2021} describe an attempt to address this problem, but primarily focus on the VSP selection problem of UAVs without taking into account the synchronization intensity control and temporal values of DTs.

\begin{figure}
\centering
\includegraphics[width=0.9\linewidth]{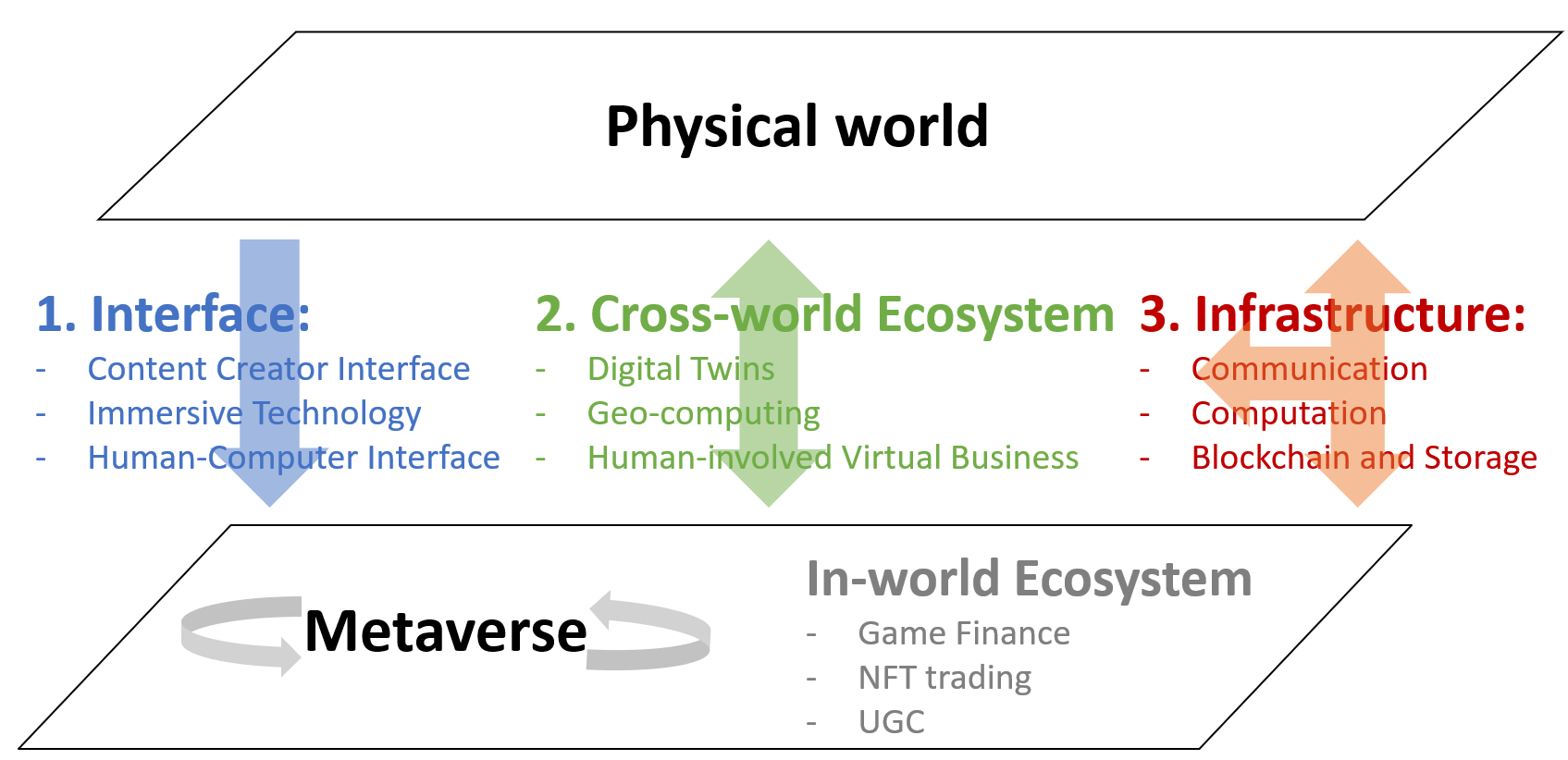}
\caption{Four components for the Metaverse \cite{hanDynamicResourceAllocation2021}}
\label{fig:arc}
\end{figure}

Regarding the Metaverse architecture, there is no consensus, e.g., a seven-layer system \cite{radoffMetaverseValueChain2021} and a three-layer architecture \cite{duanMetaverseSocialGood2021}. However, based on the Metaverse's functionality, overall, its architecture should include four aspects as shown in \cref{fig:arc}: \textit{\textbf{infrastructure}} (the fundamental resources to support the platform, such as communication\cite{luCommunicationEfficientFederatedLearning2021}, computation, blockchain\cite{yaqoobBlockchainDigitalTwins2020}, and other decentralization techniques), \textit{\textbf{interface}} (immersive technologies such as AR, VR\cite{taheriVirtualRealityDriving2017,mangianteVREdgeHow2017}, XR\cite{marrWhatExtendedRealitya}, and next generation human-brain interconnection to enrich a user's subjective sense in the virtual life), \textbf{\textit{cross-world ecosystem}} (the services that enable frequent and large amount of data transmission between the Metaverse and the physical world, to enable a \textit{convergence}  between the two worlds \cite{elsaddikDigitalTwinsConvergence2018}), and finally \textit{\textbf{in-world ecosystem}} (activities that happen only within the virtual world, e.g., transaction of the non-fungible token (NFT) \cite{mullerNFTNextStep2022}, playing games to earn Crypto (GameFi) \cite{nunleyPeoplePhilippinesAre2021}, and decentralized finance (DeFi) \cite{chenBlockchainDisruptionDecentralized2020}). See \cite{dionisio3DVirtualWorlds2013,monetaArchitectureHeritageMetaverse2020,leeAllOneNeeds2021} for more detailed discussions of the architectures and challenges faced by the Metaverse. Among these four aspects, this paper is primarily focused on the cross-world system for the Metaverse, i.e., striving for a convergence of the two worlds.

\subsection{Digital Twins (DTs) in the era of Metaverse}
DTs is not a new concept. They were proposed in 2003 \cite{grievesDigitalTwinManufacturing2014} by Grieves in his course on \enquote{product life cycle management}, which defines DTs by physical product, virtual product, and their connections. 
Later in 2012, the National Aeronautics and Space Administration (NASA) defined DTs as ``integrated multiphysical, multiscale, probabilistic simulations of an as-built vehicle or system using the best available physical models, sensor updates, and historical data'' \cite{glaessgenDigitalTwinParadigm2012}. Although these are early definitions, it is clear that a key feature of DTs is its \textit{mirroring} of physical objects, which means that \textit{continuous updates} from physical space to virtual space are needed (physical$\rightarrow$virtual), as a physical object's state changes over time. 

There is no consensus regarding the {updates} from DTs to physical objects (virtual $\rightarrow$ physical). On the one hand, the authors in \cite{fullerDigitalTwinEnabling2020,taoDigitalTwinIndustry2019,khanDigitalTwinEnabled6GVision2021} believe that there are \textit{bi-directional} updates across cyberspace and physical space for DTs (virtual $\leftrightarrow$ physical). To emphasize the bi-directional updates, the authors in \cite{fullerDigitalTwinEnabling2020} use \textit{digital models} to refer to the case where there is no update across the two spaces and use \textit{digital shadows} to refer to the case where there is a \textit{one-directional} update from physical space to the cyberspace, i.e., a change in the state of the physical object leads to a change in the digital object and \textit{NOT} vice versa. On the other hand, several other works use DTs as simulation-based only, such as \cite{grievesDigitalTwinManufacturing2014,gaborSimulationBasedArchitectureSmart2016,weyerFutureModelingSimulation2016}. That is, they treat DTs as digital shadows only (i.e., physical$\rightarrow$virtual). Meanwhile, the control of physical assets via their DTs (physical$\leftarrow$virtual) are particularly studied around another concept called cyber-physical systems (CPS) \cite{CyberPhysicalSystemsCPS}. CPS can be leveraged to support large distributed control e.g., automated traffic control and ubiquitous healthcare monitoring and delivery. 

Given the ambiguity of definitions of DTs, a number of papers have already called for  more precise clarity on the difference between the DTs, the digital shadows, and the CPS \cite{fullerDigitalTwinEnabling2020,khanDigitalTwinEnabled6GVision2021}. In this paper, we just use the term DTs to refer to the \textit{digital existence of the Metaverse}, which is \rom{1} a digital {replication} of a physical object and \rom{2} has \textit{at least} one directional information transfer from the physical world to cyberspace.
We do not distinguish whether or not there is information flow sending from a DT to a physical entity, and we leave it for each VSP to decide, based on its virtual service. For example, if virtual services are to create  simulated environments for Metaverse users to experience, such as virtual sightseeing, then bi-direction updates across worlds may not be needed; however, for an online traffic optimization service, then the control messages sent from the DTs (e.g., traffic lights) to its physical objects are important. 

Compared to smart manufacturing or Industry 4.0 in which DTs are created and studied merely for a particular type of application \cite{taoDigitalTwinIndustry2019}, the Metaverse, with a platform feature, can support different kinds of virtual business\cite{leeAllOneNeeds2021}. In addition, many brand new application scenarios are expected to emerge in the Metaverse, due to the interoperability \cite{radoffWeb3InteroperabilityMetaverse2021} (such as sharing of DTs among VSPs), human experience being important in the virtual services \cite{ARVRTechnologies2022}, and the maturity of blockchain and artificial intelligence \cite{nguyenMetaChainNovelBlockchainbased2021}. For example, the combination of NFT and DTs is considered to be the next step of NFT\cite{mullerNFTNextStep2022}, in which the value of a NFT is attached to a physical asset and the change of the asset's state affects the value of the NFT. Therefore, given a wide application of DTs in the era of the Metaverse and the necessity to synchronize the states of DTs and their physical counterparts, our generalized system model constitutes a novel approach to address synchronization issues for DTs in the Metaverse with the assistance of IoT devices for sensing the states of real entities. 

\section{System Model and Problem Formulation}\label{sec-system-model}
We consider a network that consists of \rom{1} $N$ edge devices (e.g., UAVs)  represented by the set $\mcal{N}=\{1,\ldots,n,\ldots,N\}$ and \rom{2} $M$ virtual service providers (VSPs) represented by the set $\mcal{M}=\{1,\ldots,m,\ldots,M\}$, as shown in \cref{fig-system-model}. An element of the set $\mcal{M}$ and $\mcal{N}$ is represented by $m$ and $n$, respectively. We consider UAVs as the IoT devices hereafter. 

Each of VSPs uses a set of DTs that are critical to its virtual business profit. For example, a fresh and up-to-date information of the cars on the road (traffic twin) can be critical for a city government, logistics firm, or smart goods delivery company, so that those VSPs can have correct and reliable simulation in the Metaverse to determine the optimization intervention. As time passes, DTs decline in value (e.g., due to a decrease in reliability) and as a result, the virtual business profits generated by the DTs may decrease. 

Let $\theta_m>0$ represent the value decay rate (e.g., reliability) of the DTs of VSP $m$.  We consider that the different VSPs may have different DT value decay rates, possibly due to, e.g., the types of their virtual services, computational capabilities, and the extents to which they use AI. For example, a VSP with a stronger AI algorithm and higher computation capacity may have lower value decay rates, since with the adoption of AI, VSPs can study the historical state data obtained from the physical counterparts and adjust the status of DTs based on the predictable pattern trained with AI, thereby maintaining the reliability of the DTs \cite{wangUnmannedAircraftSystem2021}. 

Let $z_m(t)\geq 0$ denote the values of the DTs for VSP $m$ at time instant $t$. We model the rate of change for the DT values, which is the first order time derivative $\dot{z}_m(t) =\rm{d}z(t)/\rm{d} t$ as follows:
\begin{equation}\label{eq-value-dynamics}
\dot{z}_m(t) =\eta_m(t) - \theta_m z_m(t), \quad m\in\mcal{M},
\end{equation}
where $\eta_m(t)$ denotes the intensity, or rate, at which the synchronization activities are carried at time $t$. One can interpret the DT value dynamics in \eqref{eq-value-dynamics} as follows: if VSP $m$ determines not to synchronize DTs at all, i.e., $\eta_m(t)\equiv 0$, then the value of DTs deteriorates at the (time independent) rate $\theta_m$. By using a positive rate of synchronization, i.e., $\eta_m(t)>0$, the VSP can slow down, or even reverse, the process of deterioration of its DTs. To simplify, we use $\bm{z}(t)=[z_m(t)]_{m\in\mcal{M}}$ to denote the vector of DT values of all VSPs and let the vector $\bm{\eta}(t)=[\eta_m(t)]_{m\in\mcal{M}}$ denote the synchronization strategies of all VSPs at each instant $t$. 

To synchronize DTs with their real counterparts, the states data of the real counterparts (hereinafter referred to synchronization data) are needed. Here, we consider a set of $N$ UAVs can be motivated to assist VSPs synchronization tasks by sensing the corresponding real entities states for the VSPs. Here, for simplicity, we consider a group of UAVs with the same type, e.g., characterized by the same sensing quality and unit energy cost \cite{shakeriDesignChallengesMultiUAV2019}. The extension to the scenario with heterogeneous types of UAVs is straightforward, as the set of UAVs can always be partitioned into multiple sub-populations so that the UAVs within a sub-population are of the same type. 

Based on synchronization requests from $M$ VSPs, each of the $N$ UAVs can select a VSP to work for. UAVs that select the same VSP are expected to collectively sense the real entities of interest to the VSP and share the incentives provided by the VSP. It is expected that when VSP $m$'s synchronization intensity increases, VSP $m$ will allocate more incentives to motivate UAVs to work for it. As such, the total incentive pool from VSP $m$ should be positively correlated with the synchronization intensity chosen by VSP $m$.

In summary, the synchronization intensity $\eta_m(t)$ affects the values of its own DTs as well as the percentage of UAVs that sense the physical twins to provide synchronization data (UAV's VSP selection distribution). Since the synchronization intensity is controlled by the VSP, we also refer to it as \textit{a control variable, a control, or a strategy} for a VSP, which is a function of the time $t$. By determining an optimal control strategy on the synchronization intensity, the VSP can affect the states of the system, including DT value states in addition to the UAV's VSP selection strategy, thereby optimizing the utility for the VSP.

To identify the optimal control strategy of the synchronization intensity and its associated DT values and the UAV's VSP selections, we focus our study on a hierarchical game formulation as follows:
\begin{itemize}
\item \textit{Lower-level Evolutionary Game}: At the lower level, we study VSP selection strategies for UAVs. Every UAV is considered to have bounded rationality, that is, to select a strategy that is satisfactory rather than optimal \cite{weibullEvolutionaryGameTheory1997}. This is to counter the case that UAV decisions are sub-optimal with the potentially incomplete information of the game (e.g., the payoffs received by other UAVs), especially when the number of UAVs is large. In this regard, we adopt and formulate the evolutionary game model in \cref{sec-lower-evoluitaonry-game} to model the strategy adaptation process of the UAVs.
\item \textit{Upper-level (Stackelberg) Differential Game} : At the upper-level, we study the optimal synchronization strategies of the VSPs. As the synchronization intensity jointly affects the synchronization data supply and the value status of the DTs, the VSPs need to devise the optimal control strategy about synchronization intensity so that the accumulated utilities discounted at the present time are maximized. We adopt a differential game approach to solve the problem in \cref{sec-upper-diff-game}. Additionally, when there is an influential VSP in the Metaverse market that has the privilege of determining strategy first, we adopt the Stackelberg differential game to solve the problem, so as to practicalize our solution. 
\end{itemize}

The notations used in the paper are presented in \cref{tab:my-table}. Arguments to a function, e.g., time variable $t$, may be omitted from expressions when there is no ambiguity.

\begin{table}[t!] 
\begin{threeparttable}
\caption{Notations used in the system model}
\label{tab:my-table}
\rowcolors{2}{gray!10}{white}
\begin{tabular}{p{1.2cm}| p{6.9cm}}
\hline
\rowcolor{gray!30}
\textbf{{Notation}} & \textbf{{Description}}   \\ 
\hline
&\cref{sec-system-model} \\
$m\in\mcal{M}$ & index of a VSP, $|\mcal{M}|=M$\\
$n\in\mcal{N}$ & index of a UAV, $|\mcal{N}|=N$  \\
$\eta_m(t),\bm{\eta}$ & synchronization intensity of VSP $m$, vector of $[\eta_m]_{m\in\mcal{M}}$ \\
$z_m(t), \bm{z}$ & values of DTs of VSP $m$, vector of $[z_m]_{m\in\mcal{M}}$\\
$\theta_m$ & time-independent value decay parameters for VSP $m$\\
$\mcal{T}$ & Time horizon $\mcal{T}=[0,T]$ \\ 
\hline
& \cref{sec-lower-evoluitaonry-game} \\
$x_m(t)$ & percentage of UAVs selecting VSP $m$ at time $t$\\
$\bm{x}$ &vector of $[x_m]_{m\in\mcal{M}}$\\
$u_m(\bm{x},\eta_m)$& utility of UAVs selecting VSP $m$\\
$d_m$ & number of DTs for VSP $m$\\
$g(\cdot)$ &weighting function used in the incentive pool\\
$R_m(t)$ & incentive pool for the VSP $m$ \\
$c_m$ & cost of UAVs in assisting VSP $m$ for synchronization \\
$\bar{u}$ & average utility of the UAVs\\
$\delta$ & learning rate \\
\hline
&\cref{sec-upper-diff-game,sec-upper-hierachical-play} \\
$b$ & average data contribution from each UAVs\\
$\alpha_m$ & unit data price for VSP $m$\\
$\beta_m$ & VSP preference of a unit increase in DTs values\\ 
$v_m$ & VSP preference of the values of the DTs\\ 
$k_m$&average data size request rate\\
$\omega_m^i$ & weights in $J_m$, $i=1,2,3,4$\\ 
$\rho$ & discounting factor\\
$J_m$& instantaneous utility function of VSP $m$\\
$J_m^i$& components for the $J_m$, $i=1,2,3,4$ \\
$\mathfrak{J}_m$ & the objective function for the VSP $m$\\
$H_m$ &Hamiltonian function of VSP $m$\\
$H_m^*$ &Maximized Hamiltonian function of VSP $m$\\
\hline
\hline
\end{tabular}
\end{threeparttable}
\end{table}

\section{Lower-level Evolutionary Game} \label{sec-lower-evoluitaonry-game}
In this section, we adopt the evolutionary game to capture the bounded rationality and dynamics of the UAVs' VSP selections. In \cref{sec-sub-lower-population-formulation}, we introduce the UAV owner's utility model and replicator dynamics that characterize the evolution of their VSP selection strategies. We prove the existence, uniqueness, and stability of the lower-level evolutionary game in \cref{sub-equilbrium-analysis}.

\subsection{UAVs Population Formulation} \label{sec-sub-lower-population-formulation}
Evolutionary game is formulated based on a set of populations with evolutionary dynamics. The various aspects of the evolutionary game are as follows:

\begin{itemize}
\item \textit{Players and Populations: } 
Each UAV $n\in \mcal{N}$ is a player of the evolutionary game. In addition, the set $\mcal{N}$ is referred to as the population of players. 

\item
\textit{Strategy: } The VSP selection is a pure strategy that can be implemented by the player. 

\item
\textit{Population States: }
Population states is the strategy distribution for the population, denoted by a vector $\bm{x}(t)=[x_m(t)]_{m\in\mcal{M}}$. The component $x_m(t)$ denotes the percentage of UAVs in the population that select VSP $m$ at the time instant $t$. As the population states are subject to $\sum_{m\in\mcal{M}} x_m(t)=1$ and therefore the state space, i.e., the set of all the possible population states, is a unit simplex $\Delta \in \bb{R}^{M-1}$.

\item
\textit{Utility Functions: } Utility function $u_m(\bm{x}(t),\eta_m)$ describes the utility that a UAV can receive given the population states $\bm{x}(t)$ and the synchronization strategy $\eta_m$. 
\end{itemize}

Each VSP $m$ has $d_m$ DTs and chooses the synchronization intensity as $\eta_m(t)$, where $\eta_m(t)\geq 0$, at time $t$. Let the time horizon for the analysis be defined as $\mcal{T}=[0,T]$. 
We refer to $\bm{\eta}(t)$, $t\in\mcal{T}$, as the control path over the time horizon $\mcal{T}$. We omit time variable $t$ for simplicity if there is no confusion. Note that determination of the optimal synchronization strategy path over $\mcal{T}$ can be further explored as an optimal control problem, which is to maximize the present value of the accumulated utility that the VSP can obtain. The extension is in \cref{sec-upper-diff-game}. In this section, we assume that the optimal synchronization intensity path has been determined by the VSPs. 

For VSP $m$, the incentive pool that it allocates to its associated UAVs is given as follows:
\begin{equation}\label{eq-incentive-pool}
R_m(t)= {\eta_m(t) d_m g(\theta_m)},
\end{equation}
where $g(\cdot)$ is a monotonically increasing function representing the weight affected by the value decay rate $\theta_m$. One can interpret \eqref{eq-incentive-pool} as both synchronization rate $\eta_m$ and the number of DTs $d_m$ have a positive correlation with the incentive pool. That is, the higher the synchronization intensity or the number of DTs that a VSP has, the more incentives the VSP should offer UAVs. Similarly, the higher the decay rate $\theta_m$ is, the more incentives a VSP should offer. We consider an affine mapping for $g(\cdot)$, e.g., $g(\theta_m)=g_0+g_1\theta_m$, where $g_1$ is a positive number to represent such a positive correlation between the decay rate and the incentive. 

Given the population states $\bm{x}$, there are $N x_m$ UAVs that select VSP $m$ to sense the data and assist its synchronization tasks. With a uniform incentive sharing scheme \cite{hanDynamicResourceAllocation2021}, each UAV can obtain the incentives with amount $\frac{R_m(t)}{Nx_m(t)}$. Let $c_m$ represent the energy cost incurred by the sensing task for VSP $m$, e.g., UAV's energy cost flying from the base to the target region and the communication cost \cite{limFederatedLearningUAVEnabled2021}. The utility received by a UAV that select VSP $m$ is as follows:
\begin{equation}\label{key}
u_m(\bm{x},\eta_m) =\frac{R_m(t)}{Nx_m(t)}-c_m=  \frac{\eta_m(t) d_m (1+\theta_m)}{N x_m(t)} - c_m.
\end{equation}

The utility information for selecting each of the VSPs at the current time $t$ can be exchanged among UAVs, e.g., at their base or device-to-device (D2D) communication in the air. As such, the UAV may adjust its VSP selection strategy at time $t+1$. The evolutionary process of the VSP selection strategy can be modeled by \textit{replicator dynamics} \cite{weibullEvolutionaryGameTheory1997}, which is a set of ordinary differential equations, given as follows:
\begin{equation}\label{eq-replicator-dynamics}
\dot{x}_{m}  =  \delta x_{m} (u_m - \bar{u}), \quad m\in\mcal{M},
\end{equation}
where $\delta$ is the learning rate of the UAVs, $\dot{x}_{m}:=dx_{m}(t)/dt$, and $\bar{u}:=\sum_{m\in\mcal{M}}x_m u_m$ denotes the average utility that a UAV population can have. Again, we omit arguments in $u_m(\bm{x}(t),\eta_m)$, and $x_m(t)$ for simplicity.

It can be seen from \eqref{eq-replicator-dynamics} that the population state $\bm{x}_m(t)$ evolves when the payoff received by a UAV is different from the population average utility. If the reward received by a device that selects VSP $m$ is higher than the average utility, i.e., $u_m>\bar{u}$, then the population state $x_{m}(t)$ increases since more UAVs adapt their strategies and select VSP $m$, i.e., $\dot{x}_{m} >0$. The evolutionary process stops when $\dot{x}_{m}=0$ for all $m\in\mcal{M}$. This is called the \textit{stationary state} or \textit{evolutionary equilibrium} (EE). The stationary states can be achieved by either $x_{m}=0$ for all VSP except one or $u_{m}=\bar{u}$, $\forall u\in\mcal{M}$. The former leads to a set of \textit{boundary stationary points}\cite{weibullEvolutionaryGameTheory1997,hanOpportunisticCodedDistributed2021} lying on a vertex of $\Delta$  and the latter leads to a set of \textit{interior stationary points}. Next, we prove that the EE uniquely exists for the lower-level evolutionary game, and is evolutionarily robust as well.

\subsection{Existence, Uniqueness, and Stability of EE }\label{sub-equilbrium-analysis}

Since the payoff is affected by the synchronization strategy $\eta_m(t)$ adopted by the VSP $m$, the evolution of a population state described by the replicator dynamics is subject to the controls of VSPs. The following theorem discusses existence and uniqueness of the solution under these controls.

\begin{theorem}\label{theorem1}
For the dynamical system defined in \cref{eq-replicator-dynamics} with initial condition $\bm{x}=\bm{x}_0$,
there exists a unique solution $\bm{x}(t)$ defined for all $t\in[0,T]$.
\end{theorem}

\begin{proof}
Let the right hand side of \eqref{eq-replicator-dynamics} denote by $f_{m}(\bm{x},\bm{\eta})$. For fixed control path $\bm{\eta}(t)$, let $\tilde{f}_{m}(\bm{x},t) := f_{m}(\bm{x},\bm{\eta} )$. Then
the differential equation \eqref{eq-replicator-dynamics} reduces to the ordinary differential equation
\begin{equation}\label{key}
\dot{x}_{m}(t)  = \tilde{f}_{m}(\bm{x},t),
\quad \forall m\in\mcal{M}.
\end{equation}
First, $\tf_m$ is continuous when $\bm{\eta}(t)$ is continuous. Next, for a fixed time $t$ and given control sequence $\bm{\eta}(t)$, the utility $u_m(t)$ is bounded, achieving maximum value when $x_m=1/d_m$ and the minimum value when $x_m=1$ for $x_m>0$. If $x_m=0$, $u_m\equiv0$ and is bounded as well. Therefore the partial derivative of $u_{m}$ w.r.t $x_{m}$, $\frac{\partial u_m}{\partial x_m}$, is bounded and thus $\frac{\partial u_m}{\partial x_q}=0$ for $q\neq m$. Therefore, $\frac{\partial \bar u}{\partial x_m} = u_m + x_m\frac{\partial u_m}{x_m}$ is bounded as well. Next, we can show that the partial derivatives of $\tf_m$ w.r.t. $x_q$ for $q\neq m$ is bounded, since $\frac{\partial \tf_m}{\partial x_q} =-\delta x_m[u_q+x_q\frac{\partial u_q}{\partial x_q}]$ and partial derivatives of $\tf_m$ w.r.t. $x_m$ is bounded given that $\frac{\partial \tf_m}{\partial x_q} =\delta(u_m-\bar{u})+\delta  x_m[\frac{\partial u_m}{\partial x_m} - \frac{\partial \bar{u}}{\partial x_m}]$. Thus $\bm{\tf}$ is bounded for all $(\bm{x},t) \in \Delta \times \bb{R}$, i.e., the Cartesian product of the state space and the one-dimensional real space.
Therefore, by the Mean Value Theorem, it can be proved that $\left| \tf_{m}(\bm{x},t) -  \tf_{m}(\bm{y},t) \right| / \left| \bm{x} - \bm{y} \right| $ is bounded for all $t\in\bb{R}$, which implies that $\tf_{m}(\bm{x},t)$ satisfies the global Lipschitz condition  \cite{engwerdaLQDynamicOptimization} and consequently a unique solution to the dynamical system exists globally.
\end{proof}

After proving that the solution $\bm{x}(t)$ exists given the initial state $\bm{x}_0$, conditioned on any control from VSPs, we now show that the solution is \textit{asymptotically stable}.

A state $\bm{x}$ is \textit{Lyapunov stable} means if no small perturbation of the state induces a movement away from $\bm{x}$.
A state $\bm{x}\in\Delta$  is \textit{asymptotically stable} if it is Lyapunov stable and all sufficiently small perturbations of
the state induce a movement back towards $\bm{x}$.  It is not hard to prove that the boundary stationary points are not stable. Therefore, we only consider interior stationary points.

\begin{theorem}
For a single UAV population, the interior stationary points of the dynamical system \eqref{eq-replicator-dynamics} are asymptotically stable.
\end{theorem}
\begin{proof}
Based on the Lyapunov direct method \cite{weibullEvolutionaryGameTheory1997}, we need to
prove the time derivative of the Lyapunov function is strictly
negative. Let $\bm{x}^*=[x_m]_{m\in\mcal{M}}$ denote the interior evolutionary
equilibrium and let $e_{m}(t):=x_{m}^*(t)-x_{m}(t)$ be the error function over time $t$. Define the Lyapunov function $V_{m}: T\rightarrow \bm{R}_{+}$ with $V_{m}=e_m^2(t)/2 $, then the time derivative of $V_{m}$ is
\begin{equation}\label{key}
\dot{V}_{m}=-e_m\delta x_m(u_m-\bar{u}).
\end{equation}
When $u_m>\bar{u}$, the population ratio $x_m$ increases, and therefore $e_m(t)>0$.  As $x_e\neq 0$, we have $\dot{V}_e<0$. When $u_m<\bar{u}$, the population ratio $x_m$ decreases, and therefore $x_m(t)<0$, which also implies $\dot{V}_e<0$. Based on the Lyapunov stability theory, the interior stationary point is asymptotically stable.
\end{proof}
See \cite{niyatoDynamicsNetworkSelection2009} for the details of the implementation of an algorithm to find the ESS with the evolutionary dynamics.

\section{Upper-level differential game for simultaneous decision-making  VSPs}\label{sec-upper-diff-game}
In the upper-level, we need to solve the optimal control problem regarding the synchronization intensity given the population state dynamics in \eqref{eq-replicator-dynamics} and the value dynamics in \eqref{eq-value-dynamics}. In this section, we consider the problem in which all VSPs are simultaneous decision makers in the game, and formulate the problem as a simultaneous differential game in \cref{sim-sub-move}. We adopt the open-loop Nash equilibrium as the solution to this game in \cref{sec-sub-open-loop-for-simudiff}.

\subsection{Problem Formulation for the Simultaneous Move}\label{sim-sub-move}
We consider that the $M$ VSPs choose their synchronization strategies at the same time, and each player competes to maximize the objective functional $\mathfrak{J}_m$, i.e., the present value of utility derived over a finite or infinite time horizon, by designing a synchronization strategy $\eta_m$ that is under the VSP's control. The choice of synchronization intensity by a player, say VSP $m$, influences \rom{1} the evolution of the UAV population states $\bm{x}(t)$, \rom{2} the value states of the DT $z_m(t)$, and \rom{3} the objective functional of the other VSPs in the set $\mcal{M}$, i.e., $\mathfrak{J}_{m'}$, $m' = 1,2,\ldots,m-1,m+1,\ldots, M$. The influence to \rom{1} and \rom{2} is captured via a set of differential equations (the system dynamics). The derivation of $\mathfrak{J}_{m}$ is given as follows.

With $x_m(t)N$ UAVs assisting VSP $m$ in sensing the current states of its real world twins and $z_m(t)$ being the current values of its DTs, the current utility rate  $J_m$ at time $t$ for VSP $m$ can be described as follows: 
\begin{equation}
J_m(\bm{x},\bm{z},\bm{\eta},t)=\omega_m^1J_m^1+\omega_m^2J_m^2-\omega_m^3J_m^3-\omega_m^4J_m^4 \label{Jm},
\end{equation}
where
\begin{align}
&J_m^1=x_mNb \alpha_m ,\quad J_m^2=\beta_m z_m d_m \label{J1J2}, \\
&J_m^3=(z_m-v_m)^2, \quad  J_m^4=(x_mNb-\eta_m d_m k_m)^2 \label{J3J4}.
\end{align}
That is, $J_m$ can be interpreted as the weighted sum of four utility components $J_m^1,J_m^2,J_m^3$, and $J_m^4$, in which $J_m^1$ and $J_m^2$ are the positive utility, and $J_m^3$ and $J_m^4$ are the disutility. $\omega_m^i\geq 0,i=\{1,2,3,4\}$ are the weight parameters to form the objective function $J_m$. We explain the utility component as follows:

\begin{itemize}
\item $J^1_m$: represents the gains generated by the the acquisition of new data of the size $x_mNb$, where $b$ represents the average amount of data that a UAV transmits to a VSP. Here, with synchonrzaiton data from the UAVs, VSP as a data supplier to the Metaverse platform, can benefit by selling the data to the Metaverse platform. The platform, as an intermediary to provide the data interoperability, can benefit the other VSPs to construct the DTs for their own use. Therefore there is a portion of revenue inflow for VSP $m$ that is linked to the data supply, or the data contribution from the UAVs. Let $\alpha_m$ denote the unit data price for the VSP $m$, then we have $J_m^1=\alpha_mx_mNb$ in \eqref{J1J2} .

\item $J^2_m$: represents the gains (e.g., virtual business profit) generated by the DTs with value of $z_m$. Here, we consider that the business is positively correlated with the DTs values. Therefore, the gains can be evaluated as $J^2_m=\beta_m z_m d_m$, where 
$\beta_m$ denotes the unit preference value that VSP $m$ has towards a unit increase in the value of the DTs. In addition, $\beta_m$ is considered to concave-upward w.r.t. $\theta_m$, e.g,, $\beta_m=e^{10\theta_m}$. This is to indicate DTs are valued more when the VSP is more sensitive to the non-updated DTs, i.e., a higher valued decay rate.
 
\item $J^3_m$: represents a penalty term, reflecting disutility when DT values are far away from the preferred values, e.g., the twins are not fresh enough (under-synchronized) or too fresh than what is needed (over-synchronized, leading to excessive synchronization cost). Let $v_m$ denote the VSP's desired values of its DTs, and $J^3$ can be defined as $(z_m-v_m)^2$ accordingly \cite{zhuPricingSpectrumSharing2014}.

\item $J^4_m$: represents the disutility caused by UAVs' insufficient data supply. As mentioned before, $d_m$ is the number of DTs of VPS $m$. 
With the synchronization intensity $\eta_m$, overall, the total amount of data that VSP $m$ requires from the UAVs devices is $d_m\eta_m k_m$, where $k_m$ is the average unit data request rate of the DT. However, since there are $x_mN$ devices that choose VSP $m$, the total data contribution to VSP $m$ is $x_mNb$ as stated earlier. The gap of $(d_m\eta_m k_m-x_mNb)$ results in disutility to the VSP $m$. For example, when an insufficient number of UAVs select VSP $m$, UAVs can complete the synchronization task at lower sampling rates \cite{liuIntelligentUAVsTrajectory2021}, resulting in lower quality synchronization data 
and affecting the utility of the VSP $m$.  
We adopt the square term to prevent the data from being over-contributed as well. 
\end{itemize}

The objective functional $\mathfrak{J}_m(\bm{\eta}) $ to be maximized for VSP $m$ is defined by the discounted cumulative payoff over the time horizon $\mcal{T}$, expressed as follows:
\begin{equation}
\begin{aligned}
\mathfrak{J}_m(\bm{\eta}) = \int_{0}^{T} e^{-\rho t} J_m(\bm{x}(t),\bm{z}(t),\bm{\eta}(t),t) \dif t \\ 
= \int_{0}^{T} e^{-\rho t} \{ 
\omega_m^1x_m(t)Nb \alpha_m+\omega_m^2z_m(t) \beta_m -\omega_m^3 \\
(z_m(t)-v_m)^2 -\omega_m^4[x_m(t)Nb-\eta_m(t) d_m k_m]^2  \} \dif t , 
\end{aligned}
\end{equation}
where $\rho\geq 0$ denotes the constant time preference rate (or discount rate) for VSPs.  $J_m(\bm{x}(t),\bm{z}(t), \bm{\eta}(t), t)$ is the instantaneous utility derived by choosing the synchronization intensity value $\bm{\eta}(t)$ at time $t$ when the current states of the game is $\bm{x}(t)$ and $\bm{z}(t)$, as explained earlier in \eqref{Jm}. 

Therefore, the optimal synchronization intensity control problem for VSP $m$ can be formulated as:
\begin{align}
\max_{\eta_m} \quad &   \mathfrak{J}_m(\bm{\eta})  \label{objective-functional}\\
\textrm{s.t.} \quad & \dot{x}_{m}(t)  =  \delta x_{m}(t) (u_m(t) - \bar{u}(t)), \quad \forall m\in\mcal{M} \label{constraint1-dynamics-x} \\
& \dot{z}_m(t) =\eta_m(t) - \theta_m z_m(t), \quad \forall m\in\mcal{M} \label{constraint1-dynamics-z}\\
  &\bm{x}(0)=\bm{x}_0, \quad \bm{z}(0)=\bm{z}_0  \label{constraint2-initial-values}\\
&\bm{x}(t)\in\Delta, \bm{z}_m(t)\geq 0, {\eta}_m(t)\geq 0,  \label{constrain2-space}
\end{align}
for $m=1,2,\ldots,M$, where column vectors $\bm{x}(0) $ and $\bm{z}(0) $ are initial states for the population states of UAVs and VSPs of DT values.

\subsection{Open-Loop Nash Solution} \label{sec-sub-open-loop-for-simudiff}
A \textbf{Nash solution} or Nash equilibrium is an $M$-tuple of synchronization strategies $\bm{\eta}=[\eta_1,\eta_2,\ldots,\eta_M]$ such that, given the opponents' equilibrium synchronization strategies, no VSP has an incentive to change its own strategy. 
Denote the synchronization strategies of VSPs other than $m$ as $\bm{\eta}_{-m}:=[\eta_1,\eta_2,\ldots,\eta_{m-1},\eta_{m+1},\ldots,\eta_M]$.
In the differential game, the Nash solution is defined by a set of $M$ admissible trajectories $\bm{\eta}^*:= [\eta_1^*, \eta_2^*,\ldots,\eta_M^*]$, which have the property that 
\begin{equation}\label{open-loop-nash-solutions}
\mathfrak{J}_m(\bm{\eta}^*) = \max_{\eta_m} \mathfrak{J}_m(\eta_1^*,\ldots,\eta_{m-1}^*,\eta_m, \eta_{m+1}^*,\ldots, \eta_M^*), 
\end{equation}
for $m=1,2,\ldots,M$. 

Next, we adopt the \textbf{open-loop solutions} for the above Nash differential game. The open-loop Nash solution to the optimal control problem refers to the case where the control paths are functions of time $t$ only, satisfying \eqref{open-loop-nash-solutions}.  
For simplicity, hereafter, we use a column vector $\bm{y}$ to represent the system states $\bm{x}$ and $\bm{z}$, i.e., $\bm{y}=[x_1,x_2,\ldots,x_M,z_1,z_2,\ldots,z_M]^T$. Then, the constraints defined in \eqref{constraint1-dynamics-x} to \eqref{constrain2-space} can be replaced by the following conditions: 
\begin{align}
&\dot{\bm{y}}(t) = [\dot{x}_1,\ldots,\dot{x}_M,\dot{z}_1,\ldots,\dot{z}_M]^{\top} \label{dynamics-y} \\
&\bm{y}(0) = [\bm{x}(0)^{\top}, \bm{z}(0)^{\top} ]^{\top} \label{initial-y}\\
&\bm{y}(t)\in \mathcal{Y}:= \Delta \times \bb{R}_+^M, \quad \eta_m(t)\in\bb{R}_+  \label{space}.
\end{align}
This means that the process to solve the open-loop Nash solution is to solve the optimal control problem, $\mathscr{P}1$, defined by
\begin{equation}
\begin{aligned}
\max_{\eta_m} \quad & \mathfrak{J}_m({\eta_m,\bm{\eta}_{-m}^*})  \label{optimal-control-NE}\\
\textrm{s.t.} \quad & \text{\cref{dynamics-y,initial-y,space}},
\end{aligned}
\end{equation}
for $m=1,2,\ldots,M$. 
To solve $\mathscr{P}1$, we first define a (current-value) Hamiltonian function $H$ as follows 
\begin{equation}\label{Ham}
\begin{aligned}
H_m(\bm{y},\eta_m,\bm{\lambda}_m,t) =  J_m(\bm{y},\eta_m,\bm{\eta}_{-m}^*,t) + \bm{\lambda}_m \dot{\bm{y}},
\end{aligned}
\end{equation}
for $m=1,2,\ldots,M$. The domain of $H_m$ is the set $\{(\bm{y},\eta_m,\bm{\lambda}_m,t) | \bm{y}\in \mathcal{Y}, \eta_m\in\bb{R}_+, \bb{\lambda}_m \in \bb{R}^{2M} ,t\in \mcal{T}\}$. Here, the row vector $\bm{\lambda}_m=[\lambda_{m1},\lambda_{m2},\ldots,\lambda_{m2M}]$ is called the (current-value) adjoint variable or costate variables. Therefore, the maximized Hamiltonian function $H^*: \mathcal{Y}\times\bb{R}^{2M}\times\mcal{T}\rightarrow \bb{R}$ is 
\begin{equation}\label{key}
H_m^*(\bm{y},\bm{\lambda}_m, t) = \max\{ H_m(\bm{y},\eta_m,\bm{\lambda}_m,t) | \eta_m\geq 0 \}.
\end{equation}  
A necessary and sufficient condition for the optimal control is given by the augmented maximum principle, stated as in \cref{theorem-1}. See \cite{docknerDifferentialGamesEconomics2000} for the detailed proof.
\begin{theorem}\label{theorem-1} 
Consider an optimal control problem $\mathscr{P}1$ and define the Hamiltonian function $H_m$ and the maximized Hamiltonian function $H_m^*$ as above. The state space $\Theta$ is a convex set and the scrap value function $S$ is continuously differentiable and concave (note that $S\equiv 0$ in $\mathscr{P}1$). If there exists an absolutely continuous function $\bm{\lambda}_m: [0,T]\rightarrow \bb{R}^{2M}$ for all $m\in\mcal{M}$, such that the maximum condition
\begin{equation}
\label{maximum-condition-eta}
H_m(\bm{y},\eta_m^*,t)  = H_m^*(\bm{y},\bm{\lambda}_m ,t) ,
\end{equation}
the adjoint (costate) equation
\begin{equation}
\dot{\bm{\lambda}}_{m} =  \rho \dot{\bm{\lambda}}_{m} - \frac{\partial H_m^*(\bm{y},\bm{\lambda}_m ,t) }{\partial \bm{y}}  \label{adjoint},
\end{equation}
and the transversality condition 
\begin{equation}\label{tranversality}
\bm{\lambda}_m(T) = S'(\bm{y}(T)) = 0
\end{equation}
are satisfied, and such that the function $H_m^*$ is concave and continuously differentiable w.r.t. $x$ for all $t\in\mcal{T}$, then $\eta_m(\cdot)$ is an optimal control path. If further the set of feasible controls does not depend on $\bm{y}$ (which is true for $\mathscr{P}1$ as $\eta_m\in\bb{R}_+$), \eqref{adjoint} can be replaced by 
\begin{equation}\label{adjoint2}
\dot{\bm{\lambda}}_{m} =  \rho {\bm{\lambda}}_{m} - \frac{\partial H_m(\bm{y},\eta_m^*,t) }{\partial \bm{y}}.
\end{equation}
\end{theorem}
Note that {${\partial H_m(\bm{y},\eta_m^*,t)  }/{\partial \bm{y}}$ is a row vector, following \cite{sethiOptimalControlTheory2019} that the derivative of a real-valued function w.r.t. a vector (no matter a column vector or a row vector) is a row vector.}
Furthermore, to solve \eqref{maximum-condition-eta}, we notice that the function $H_m$ in $\mathscr{P}1$ is strictly concave w.r.t. $\eta_m$. Therefore, we can instead solve $\eta_m^*$ by the first order optimality conditions, given as follows:
\begin{equation}\label{solved-u}
\left. \frac{\partial H_m(\bm{y},\eta_m,t)  }{\partial \eta_m} \right |_{\eta_m= \eta^*_m }= 0   .
\end{equation}

After $\eta_m^*$ is solved by \eqref{solved-u}, a boundary value problem of a system of ordinary differential equations can be defined by $\dot{\bm{y}}$ in \eqref{dynamics-y}, $\dot{\bm{\lambda}}$ in \eqref{adjoint2}, together with their boundary values defined in \eqref{initial-y} and \eqref{tranversality}.  The states for this new dynamic systems are $\bm{y}$ and $\bm{\lambda}$, which can be numerically solved using \textit{bvp4c} in Matlab, or \textit{scipy.integrate.solve\_bvp} in python.

\section{Upper-level differential game for VSPs with hierarchical decision making process}\label{sec-upper-hierachical-play}
After obtaining the solution for the simultaneous play VSPs in \cref{sec-upper-diff-game}, we now consider a more complicated realistic case, in which some VSPs are allowed by the Metaverse to choose their synchronization strategies earlier than the other VSPs. We refer to such VSPs as \textit{leaders}. The VSPs that observe leaders' strategies and then make their decisions are called \textit{followers}. To model this sequential strategic interaction among the VSPs, or a hierarchical play, we adopt the Stackelberg differential game. 
\subsection{Problem Formulation}
We use $\mcal{L}$ to denote the set of leaders and $\mcal{F}$ the set of followers, such that $\mcal{L}\cap \mcal{F}=\emptyset$ and $\mcal{L}\cup\mcal{F}=\mcal{M}$.  We use $\bm{\eta}^L=[\eta_i^L]_{i\in\mcal{L}}$ to denote synchronization strategy of leaders and $\bm{\eta}^F=[\eta_m^F]_{m\in\mcal{F}}$ synchronization strategy of followers, so as $\bm{\eta}^{L*}$ and $\bm{\eta}^{F*}$ for the optimal ones. At time $0$, the leaders announce the synchronization strategy path $\bm{\eta}^L(t)$. The followers, taking those synchronization strategy paths as given, choose their synchronization strategies $\bm{\eta}^F(t)$ so as to maximize their objective functional. 

\subsubsection{Followers' Problem $\mathscr{P}_F$}
Given the leader's optimal synchronization strategy paths $\bm{\eta}^{L}$, the followers problem $\mathscr{P}_F$, is the same as $\mathscr{P}_1$ defined in  \cref{sec-upper-diff-game}. In particular, for a follower $m\in\mcal{F}$, the optimal control problem is 
\begin{equation}
\begin{aligned}
\max_{\eta_m^F} \quad & 
\mathfrak{J}_m({\eta_m^F,\bm{\eta}^{L},\eta_{{\cal{F}} \backslash \left\{m\right\}}^{F*}}) 
 \label{optimal-control-NE}\\
\textrm{s.t.} \quad & \text{\cref{dynamics-y,initial-y,space}},
\end{aligned}
\end{equation}
where $\mcal{F}\setminus m :=\{i \in \mcal{F}, i\neq m\}$ is the set of followers other than VSP $m$. 
For the follower $m\in\mcal{F}$, its Hamiltonian function, denoted by $H_m^L$, is the same as \eqref{Ham}, that is $H_m^L = H_m$. 
Then, the optimal synchronization strategy $\eta_m^{F*}$ for the follower $m$ is equivalent to $\eta_m^*$, which is the solution to \eqref{solved-u}, and adjoint equations $\dot{\bm{\lambda}}_m$ satisfies \eqref{adjoint2}. Due to the strict concavity of $H_m^F$, $\eta_m^{F*}$ can be uniquely determined by  \eqref{solved-u}, as a function of $\bm{y}$, $\bm{\eta}^{L}$,  $[{\eta}^{F*}_{j}]_{j\in\mcal{F}\setminus m}$ and $t$, for all $m\in\mcal{F}$. That is, we can write
\begin{equation}\label{eq1}
\eta_{m}^{F*}=\mathsf{\tilde{g}}_m(\bm{y},\bm{\lambda}_m,\bm{\eta}^{L},\bm{\eta}^{F*}_{\mcal{F}\setminus \{m\}},t), \quad \forall m\in\mcal{F},
\end{equation}
which can be further simplified by substituting $[{\eta}^{F*}_{j}]_{j\in\mcal{F}\setminus m}$ based on \eqref{eq1} into $\mathsf{\tilde{g}}_m(\cdot)$. Therefore, we can express $\eta_{m}^{F*}$ as a function of $\bm{y},\bm{\lambda}_m$ and $\bm{\eta}^{L}$ as follows: 

\begin{equation}\label{follower-optimal-u}
\eta_{m}^{F*}=\mathsf{g}_m(\bm{y},\bm{\lambda}_m,\bm{\eta}^{L},t), \quad \forall m\in\mcal{M},
\end{equation}
or the vector function $\bm{\eta}^{F*}=\mathsf{g}(\bm{y},\Lambda,\bm{\eta}^{L},t)$, where $\Lambda:=[\bm{\lambda}_m]_{m\in\mcal{F}}$ represents all the adjoint (costate) variables in $\mathscr{P}_F$.

Substituting \eqref{follower-optimal-u} into \eqref{adjoint2}, we obtain
\begin{equation}\label{lambda2}
\dot{\bm{\lambda}}_m=\rho\bm{\lambda}_m -  \frac{\partial H_m(\bm{y},\mathsf{g}_m(\bm{y},\bm{\lambda}_m,\bm{\eta}^{L},t),t) }{\partial \bm{y}}, \quad  m\in\mcal{F}.
\end{equation}
\cref{dynamics-y,initial-y,space,tranversality,lambda2,follower-optimal-u} characterize the follower's best response to the leaders control path $\bm{\eta}^{L*}$.

\subsubsection{Leaders' Problem $\mathscr{P}_L$} 
As for the leaders' problem $\mathscr{P}_L$, for any leader $i\in\mcal{L}$, it knows the followers' best responses. Therefore, different from the simultaneous differential game, the system dynamics additionally include the adjoint equations of the followers' problem. Again, similar to the follower's game, among leaders, we obtain the Nash equilibrium. As such, given the best responses of all the followers, and the other opponent leaders play their strategy $\bm{\eta}^{L}_{\mcal{L}\setminus i}$, the optimal control problem for the leader $i\in\mcal{L}$ is formulated as follows: 
\begin{equation}\label{leader1}
\begin{aligned}
\max_{\eta_i^L} \quad & \mathfrak{J}_i(\eta_i^L,\bm{\eta}^{F*},\bm{\eta}^{L}_{\mcal{L}\setminus \{i\} }) \\
\textrm{s.t.} \quad & \text{\cref{dynamics-y,initial-y,space,tranversality,lambda2,follower-optimal-u}},
\end{aligned}
\end{equation}
where $\bm{\eta}^{F*}=\mathsf{g}(\bm{y},\Lambda,\bm{\eta}^{L},t)$ stated before. Then, we define the Hamiltonian function for leader $i$ as
\begin{multline}\label{leader-ham}
H^L_i(\bm{y},\Lambda,\bm{\eta}^{L}, \bm{\psi}_i,\bm{\phi}_i ,t)={J}_i(\eta_i^L,\mathsf{g}(\bm{y},\Lambda,\bm{\eta}^{L},t),t),\bm{\eta}^{L}_{\mcal{L}\setminus i})\\
+\bm{\psi}_i\dot{\bm{y}}+\sum_{m\in\mcal{F}}\bm{\phi}_{mi}\bm{\dot{\lambda}}_m^T, 
\end{multline}
for $i\in\mcal{L}$, where row vector $\bm{\psi}_i=[\psi_{ij}]_{j=1}^{2M}$ is the adjoint variables for the states $\bm{y}$, row vector $\bm{\phi}_{mi}=[\phi_{mij}]_{j=1}^{2M}$ is the adjoint variables for the adjoint variables $\bm{\lambda}_m$, and $\bm{\phi}_i=[\bm{\phi}_{mi}]_{m\in\mcal{F}}$. Note that the last term, $\bm{\phi}_{mi}\bm{\dot{\lambda}}_m^T$, is an inner product as $\bm{\dot{\lambda}}_m$ is a row vector, and $(\cdot)^T$ is the transpose operation. 

We then have the optimality conditions (again applying \cref{theorem-1})
\begin{align}
&\frac{\partial H^L_i(\bm{y},\Lambda,\bm{\eta}^{L}, \bm{\psi}_i,\bm{\phi}_i ,t)}{\partial \eta_i^L}=0\\
&\dot{\bm{\psi}}_{i}=\rho\bm{\psi_{i}}-\frac{\partial H^L_i(\bm{y},\Lambda,\bm{\eta}^{L}, \bm{\psi}_i,\bm{\phi}_i ,t)}{\partial \bm{y}} \label{dotpsi}\\
&\dot{\bm{\phi}}_{mi} = \rho \bm{\phi}_{mi}-\frac{\partial H^L_i(\bm{y},\Lambda,\bm{\eta}^{L}, \bm{\psi}_i,\bm{\phi}_i ,t)}{\partial\bm{\phi}_{mi}} \label{dotphi}, 
\end{align}
for $i\in\mcal{L}$ and $m\in\mcal{F}$. 

Different to the follower's game, there is one more transversality condition for an adjoint variable in the leader's problem, which depends on the costate variables in the follower's problem \cite{docknerDifferentialGamesEconomics2000}. Let $\lambda_{mj}\in \Lambda$ where $j=1,2,\ldots,2M$ and $m\in\mcal{F}$ be a co-state variable in the follower's problem, and $\phi_{mij}(t)\in\bm{\phi}_{mi}$ be the costate variable of $\lambda_{mj}$. Then the costate variable $\lambda_{mj}(t)$ is called \textit{uncontrollable} by the leader, if $\lambda_{mj}(0)$ is independent of the leader $i$'s control path $\eta_i$, e.g., a function of time $t$ only. Otherwise, it is said to be \textit{controllable}. For those leaders' uncontrollable states, additional transversality conditions are needed, by setting $\phi_{mij}(0)=0$.

Similarly, the concavity of the leader's Hamiltonian function $H_i^L$ in \eqref{leader-ham} ensures $\eta_i^{L*}$ can be uniquely expressed as a function of $\bm{y},\Lambda,
\bm{\eta}_{\mcal{L}\setminus \{i\}}, \bm{\psi}_i,\bm{\phi}_i,t$ for all $i\in\mcal{L}$. As such, after simplification, $\eta_i^{L*}=\mathsf{h}_i(\bm{y},\Lambda, \bm{\psi}_i,\bm{\phi}_i,t)$, or the vector function $\bm{\eta}^{L*}=\mathsf{h}(\bm{y},\Lambda,\Psi,\Phi,t)$, where $\Psi=[\bm{\psi}]_{i\in\mcal{L}}$ and $\Phi=[\bm{\phi}_i]_{i\in\mcal{L}}$. By backward induction, namely, substituting $\bm{\eta}^{L*}=\mathsf{h}(\bm{y},\Lambda,\Psi,\Phi,t)$ into \cref{follower-optimal-u,dotphi,dotpsi}, we can obtain the dynamics of the system states in \cref{dynamics-y,lambda2,dotphi,dotpsi} with only the states and time variable, i.e., $\bm{y},\Lambda,\Psi,\Phi$ and $t$. Together with the boundary conditions for the system states, a two point boundary value problem is defined, which can be solved numerically as stated earlier in \cref{sec-upper-diff-game}.

The implementation is given in \cref{algo}. The complexity is as follows. Given that there are $M$ players in the game, the population states $\bm{x}$ is of dimension $M-1$, and the DTs' value states $\bm{z}(t)$ is of dimension $M$. Thus, the overall number of system states is linear in $M$ so are the co-states variables. Therefore, it is clear that steps $8-9$ are of complexity $O(M)$ and step $4$ has complexity $O(1)$. Therefore, the overall complexity is $O(M^2)$.

\begin{algorithm}
	\caption{Implementation of Dynamic Hierarchical Framework}
	\label{algo}
	\begin{algorithmic}[1]
		\renewcommand{\algorithmicrequire}{\textbf{Input:}} 
		\renewcommand{\algorithmicensure}{\textbf{Output:}}
		\REQUIRE VSPs' utility parameters, UAVs' parameters, and system parameters. Initialize $\bm{x}(0),\bm{z}(0)$. 
		\ENSURE  optimal control path $\bm{\eta}^*(t)$ and its associated system states $\bm{x}(t)$ and $\bm{z}(t)$

		\STATE {\emph{\underline{Lower-level Evolutionary Game}}}
		\STATE Compute $u_m(\bm{x},\bm{\eta})$ for all $m\in\mcal{M}$ and $\bar{u}$
		\FOR{$m\in\mcal{M}$}
		\STATE Derive the population dynamics given in \cref{eq-replicator-dynamics}
		\ENDFOR
		
		\STATE {\emph{\underline{Upper-level Differential Game}}}
		\FOR{$m\in\mcal{M}$}
		\STATE Obtain the objective functions $\mathfrak{J}_m$ in \cref{Jm} and the Hamiltonian function $H_m$
		\STATE Derive the dynamics for the co-states variables, namely \cref{adjoint2} for simultaneous play and \cref{adjoint2,dotphi,dotpsi} for hierarchical play.
		\ENDFOR 
		\STATE Solve the boundary value problems to return $\bm{\eta}^*(t)$
	\end{algorithmic}
\end{algorithm}

\section{Performance Evaluation}\label{sec-experiments}
In this section, we examine and validate the theoretical findings presented in the previous section. First, we numerically demonstrate the existence, uniqueness, and stability of equilibrium obtained in the lower-level evolutionary game (ESS). We then conduct sensitivity analyses by varying a series of system parameters, including the number $M$ of VSPs, learning rate $\delta$, decay rate $\theta_m$, and discount rate $\rho$. Finally, we compare the results obtained by simultaneous differential game, Stackelberg differential game, and the static Stackelberg game.  Parameters used in the simulation experiment are listed in \cref{tab-simulation-parameters}.

\begin{table}[] 
\centering
\caption{Simulation Parameters}
\label{tab-simulation-parameters}
\begin{tabular}{p{6cm} | p{1.8cm}}
\hline
\hline
\textbf{{Parameters}} & \textbf{{Values}}   \\ 
\hline
Total number of VSPs $M$& $[2,4]$\\
UAV population size $N$ & $[350,500]$ \\
Learning rate $\delta$ &  $[0.01,0.1]$\\
Discount rate $\rho$ & $[0.05,0.2]$\\
Number of DTs $d_m$ & $[50,120]$\\
Digital twins value decay rate $\theta_m$ & $[0.5,1]$\\
Weight parameters $w_i,i=1,2,3,4$ & $[0.001,1]$\\
Data price $\alpha_m$ & $0.1$\\
Data contribution from each UAV each time $b$ & $[0.1,1]$ Mb\\
Average data size request rate $k_m$ & $[0.1,0.5]$ Mb \\
VSP’s desired DT values $v_m$& $60$ \\
Parameters for the affine mapping $g(\cdot)$, $g_0,g_1$ & $1$ \\
\hline
\hline
\end{tabular}
\end{table}

\subsection{VSPs are Simultaneous Decision makers}\label{sub-section-experiment-subsecA}
\begin{figure}
     \centering
     \begin{subfigure}[b]{\linewidth}
         \centering
         \includegraphics[width=\linewidth]{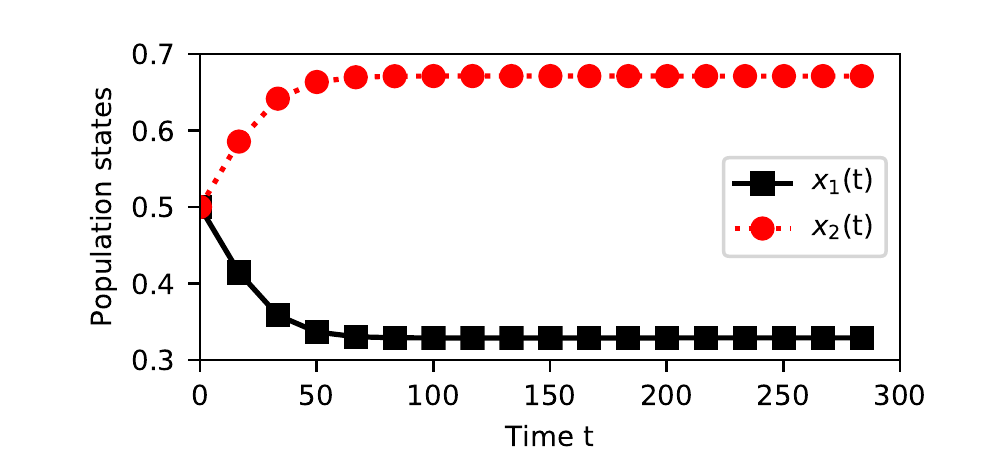}
         \caption{Evolution of UAV population states shows that the unique equilibrium in the lower-level evolutionary game for $2$ simultaneous decision-making VSPs}
         \label{fig-states1}
     \end{subfigure}
	\begin{subfigure}[b]{\linewidth}
	\centering
	\includegraphics[width=\linewidth]{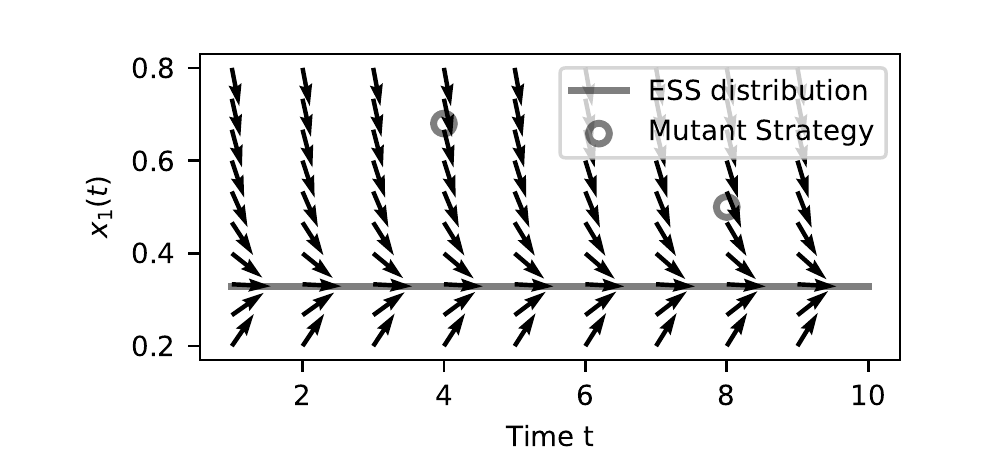}
	\caption{The direction field of the replicator dynamics shows the evolutionary stability of the equilibrium in the lower-level game.} 
	\label{fig-df_time}
	\end{subfigure}
\caption{Existence, uniqueness, and stability of the ESS in the lower-level evolutionary game}
\end{figure}
\subsubsection{Uniqueness and existence of the ESS} \label{exp-existence}
We first consider the case that VSPs simultaneously determine their synchronization strategies. We first consider a representative case with $2$ VSPs in the Metaverse, and the DT value decay rates for $2$ VSPs are $\theta_1=0.05$ and $\theta_2=0.1$. Both VSPs have $80$ DTs and the number of UAVs is $500$. The learning rate is $\delta=0.05$ and the discount rate is $\rho=1$ for both VSPs. The initial values of the DTs are $40$ for both VSPs. The time horizon $\mcal{T}=[0,300]$. The initial population states is given as $\bm{x}(0)=[0.5,0.5]$ (i.e., without prior knowledge, the chance of selecting each VSP is the same among the UAV population at time $0$). 

To examine the existence and uniqueness of the lower-level evolutionary game, we plot the trajectories of the population states over time in \cref{fig-states1}, which indicates the percentage of UAVs selecting each VSP over time. It can be seen that, $x_2(t)$ increases steadily over time, while $x_1(t)$ decreases. After a few iterations, the population states $\bm{x}(t)$ reach an equilibrium state $\bm{x}=[0.33, 0.67]$. This numerically demonstrates the unique existence of the equilibrium in the lower-level evolutionary game. 

To further understand the reason why $x_2(t)$ increases over time, we plot the synchronization strategies of both VSPs in \cref{fig-control1}. 
We observe that on average VSP $2$ adopts a higher synchronization rate than that of VSP $1$. 
Because VSPs' synchronization strategy is positively correlated with the incentive pool offered to the UAVs, VSP $2$ can issue more rewards to UAVs, and therefore UAVs that selected VSP $1$ initially may switch their VSP selection to VSP $2$ so as to enjoy higher payoffs, thereby increasing the value of $x_2(t)$ over time. However, as more UAVs select VSP $2$, the average reward that a single device can receive decreases. Therefore, $x_2(t)$ stops increasing after some rounds of iterations and remains stable, i.e., reaches an equilibrium state, after which UAVs have no incentive to change their strategies.

\subsubsection{Digital twin value states} \label{exp-zvalue}
After studying the evolution of the population states in the lower-level evolutionary game $\bm{x}(t)$, we continue the analysis of the system state by plotting the evolution of DT value states, $\bm{z}(t)$, in \cref{fig-twin1}. It can be seen that the DT values of both VSPs increase from their initial value of $40$ towards the preferred threshold $v_m=60$, after several iterations. In addition, we observe that the DT values of VSP $1$ increase faster than those of VSP $2$, in part due to the lower decay rate of VSP $1$. The figure demonstrates the validity of our proposed solution, \textit{IoT-assisted Metaverse synchronization}. In particular, to increase and preserve DT values at a certain preferred level, IoT can assist VSPs in collecting fresh data with respect to (w.r.t.) their real entities. The improved DT values result in a higher quality virtual business for VSPs as well as improving user experience. 

\begin{figure}[]
\centering
     \begin{subfigure}[b]{\linewidth}
         \centering
         \includegraphics[width=\linewidth]{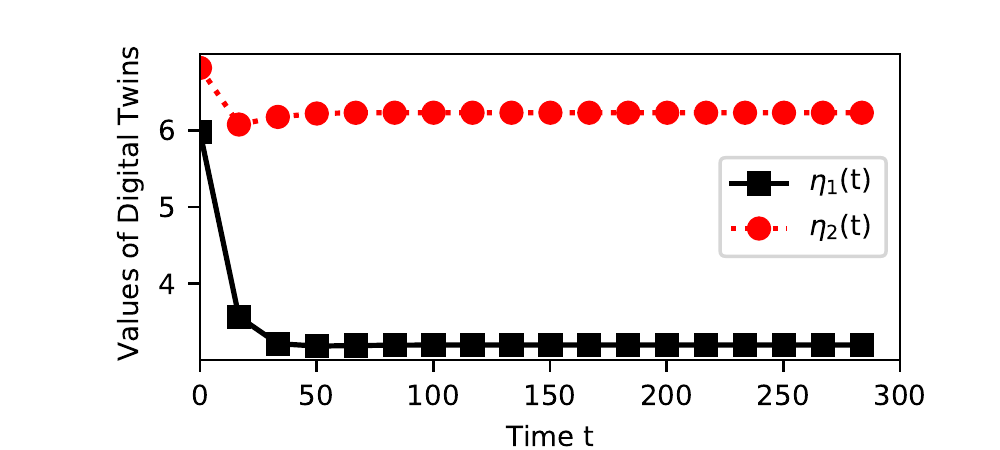}
         \caption{Trajectories of synchronization intensity, chosen by each VSP}
         \label{fig-control1}
     \end{subfigure}
\centering
     \begin{subfigure}[b]{\linewidth}
         \centering
		\includegraphics[width=\linewidth]{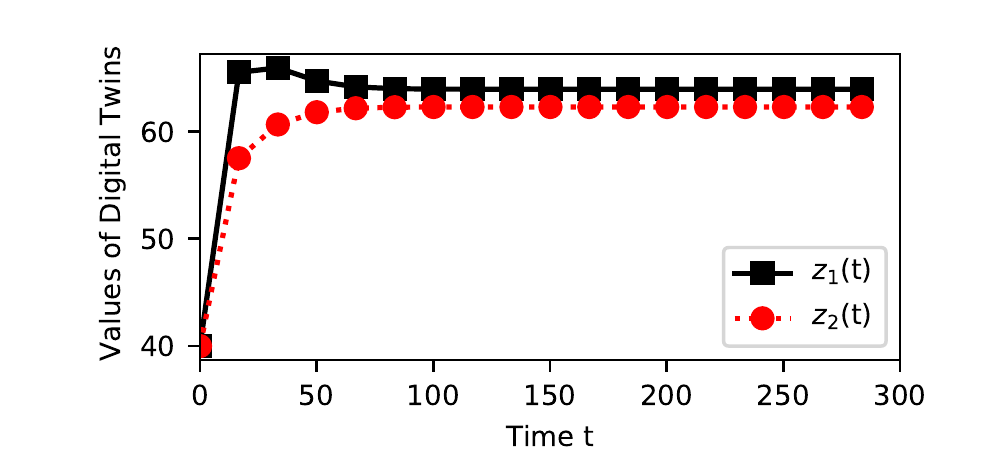}
		\caption{Trajectories of the DTs' value states}
		\label{fig-twin1}
     \end{subfigure}
\caption{Trajectories of controls and other system states in the simultaneous differential game with $2$ VSPs}
\end{figure}

\subsubsection{Stability of equilibrium in the lower-level evolutionary game}
Having demonstrated the ESS's unique existence in \cref{exp-existence}, we now investigate the stability of the ESS, namely, if the mutant strategy in the UAV population can evolve towards the ESS, given a small perturbation around the equilibrium point. We only examine $x_1(t)$ as $x_2(t)=1-x_1(t)$. We vary $x_1(0) \in [0.2,0.8]$ in step size of $0.067$ while fixing the digital value states $\bm{z}(t)$, the same as their equilibrium in the last section (\cref{exp-zvalue}). Figure \ref{fig-df_time} shows the the evolution direction of $x_1(t)$ over time. Clearly, the population states converge to the ESS ($x_1=0.33$) from any initial strategy distribution over the UAV population. In addition, mutant strategies could be eliminated by the adaptive process of VSP selection, which demonstrates the robustness and stability of the ESS in the lower-level game.

\subsection{Sensitivity Analysis}
We now study the sensitivity of the results to changes in system parameter values, including the number of UAVs $M$, learning rate $\delta$, and decay parameter $\theta_m$. 

\subsubsection{Varying $M$, the number of VSPs in the Metaverse} \label{exp-varyM}
We first consider a case with a larger value $M=4$. The parameters for VSPs $1$ and $2$ are identical to those in \cref{sub-section-experiment-subsecA}.The parameters of the additional two VSPs are $d_3=80$, $d_4=80$, $\theta_3 =0.15$, and $\theta_4=0.2$. The initial population states is $\bm{x}(0)=[0.25,0.25,0.25,0.25]$. The initial DT values are $\bm{z}(0)=[40,40,40,40]$.  

Figure \ref{fig-sens-4vsp} shows the trajectories of the population states $\bm{x}(t)$, synchronization intensities $\bm{\eta}(t)$, as well as the their DT values $\bm{z}(t)$. In general, we observe a similar pattern to the case of $ M=2$, namely, all trajectories converge in the long run. This means the dynamic interactions between the VSPs and UAVs become stable after several rounds of iterations. 
In particular, for $\bm{x}$, we observe that $x_4(t)$ increases over time, as VSP $4$ increases its synchronization intensity in response to the higher decay rate of its DTs. As a result, higher incentives (as positively correlated with the synchronization intensity) can attract more UAVs to aid the collection of synchronization data for VSP $4$. However, as shown in the second subfigure of \cref{fig-sens-4vsp}, DTs from VSP $1$ are of the highest values over time, even though its UAV selection is the lowest. 
The reason is that the decay rate of VSP $1$ is the lowest. With minimal provision of synchronization data, the values of its DTs can still be maintained. 
In contrast, though the selection of VSP $4$ is the highest among the UAV population, it is still not enough to improve the DT values of VSP $4$ to reach the threshold, $60$. 
\begin{figure}[]
\centering
\includegraphics[width=0.8\linewidth]{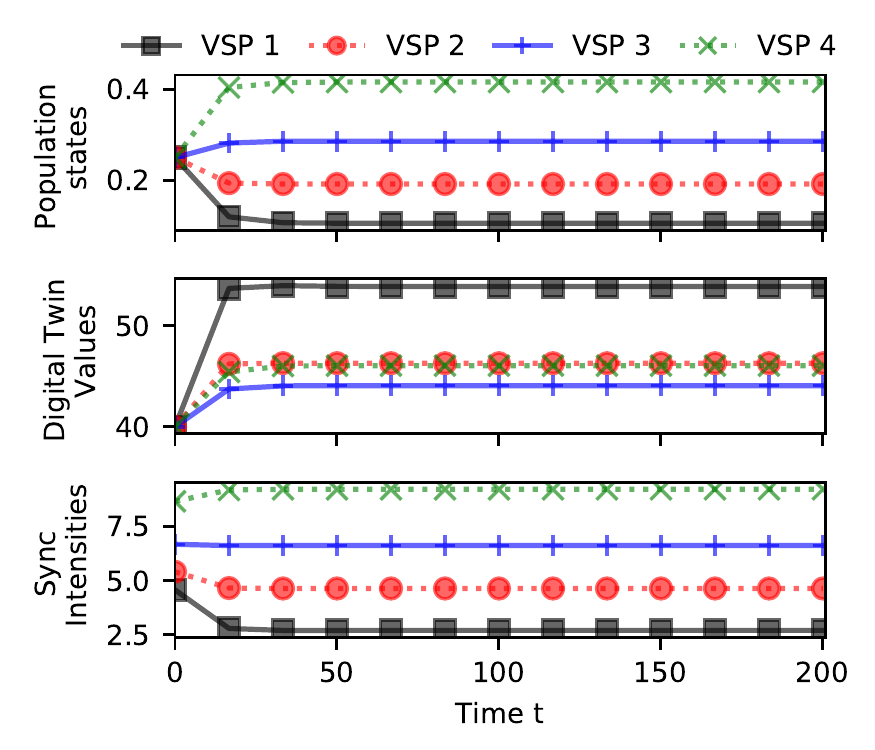}
\caption{Trajectories of population states $\bm{x}(t)$, DT values $\bm{z}(t)$, and synchronization intensities $\bm{\eta}(t)$ for the simultaneous differential game with $M=4$, exhibiting similar patterns as those of previous experiments with $M=2$.} 
\label{fig-sens-4vsp}
\end{figure}

\subsubsection{Varying $\delta$, learning rate in the evolutionary dynamics}

\begin{figure}[]
\centering
     \begin{subfigure}[b]{\linewidth}
         \centering
	\includegraphics[width=0.82\linewidth]{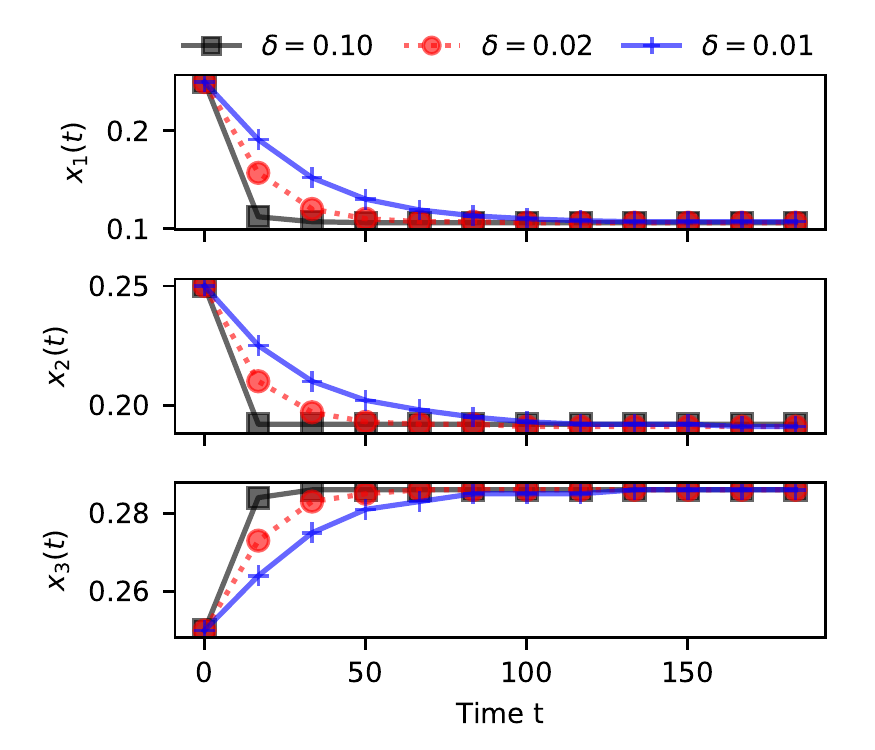}
	\caption{Population states}
	\label{delta_x}
     \end{subfigure}
\centering
     \begin{subfigure}[b]{\linewidth}
         \centering
		\includegraphics[width=0.8\linewidth]{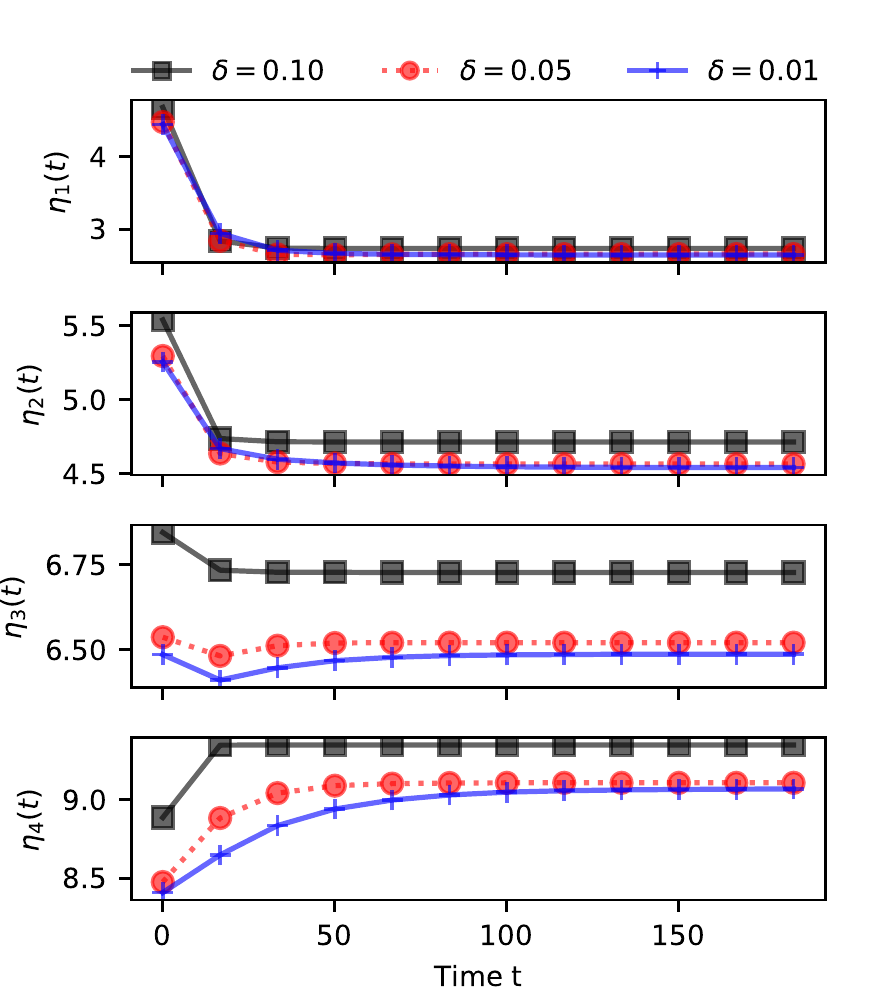}
		\caption{Sync intensities}
		\label{delta_eta}
     \end{subfigure}
\caption{Sensitivity analysis of learning rate $\delta$}
\label{exp-delta}
\vspace{-0.1cm}
\end{figure}

Figure \ref{exp-delta} shows the effect of the learning rate $\delta$ on the simultaneous differential game of the $4$ VSPs. The experiment parameters are the same as presented in \cref{exp-varyM}, except for $\delta\in \{0.01,0.02,0.1\}$. Figure \ref{delta_x} shows the trajectories of the population states $x_1(t)$, $x_2(t)$, and $x_3(t)$ under different values of $\delta$. Clearly, $\delta=0.1$ gives the fastest convergence speed of the population states in the lower-level evolutionary game, whereas $\delta=0.01$ gives the slowest speed. The reason is that the learning rate represents the frequency of population strategy adaptation, such as the percentage of UAVs adjusting their VSP selection at each decision epoch, which controls the speed of strategy adaptation in the lower-level game. 

Figure \ref{delta_eta} shows the trajectories synchronization intensities $\bm{\eta}(t)$ of four VSPs under various values of learning rate in the lower-level game. First, we observe that there is a stable strategy for all the VSPs after a few iterations for any value of $\delta$. In addition, similar to the findings in \cref{delta_x}, we can also observe that the lower the value of $\delta$, the longer it takes to achieve a stable control strategy. 
However, unlike $\bm{x}(t)$ in \cref{delta_x} that it always converge to the same equilibrium state under various $\delta$, we observe that the equilibrium synchronization intensities are at different values under the different values of $\delta$, especially for VSPs with higher decay rates, such as VSPs $3$ and $4$ in the experiment. 
The reason is that players (VSPs) in the upper-level game make decisions (synchronization intensity) based on discounted accumulative payoffs, and that recent payoffs carry greater weight on cumulative payoffs. Note that an even selection of VSPs in the UAV population does not favor VSPs $3$ and $4$ because they may require more UAVs to assist the task due to higher decay rates. Therefore, a small learning rate in the lower-level game can make VSPs $3$ and $4$ stay in the unfavored position for a longer time and decrease the overall discounted payoff. Subsequently, VSPs may decide to reduce their synchronization intensity in response to this situation. 

In contrast, selections of VSPs by UAVs modeled by the replicator dynamics in evolutionary games, are myopic. UAVs only target payoffs in the infinitesimal vicinity of the present time without discounting those future payoffs or referring to the long-term memory. 
With such a myopic nature, population states can always reach the same set of equilibrium states given the varying learning rates.

\subsubsection{Varying $\theta_m$, the decay rate of a VSP}
The impact of the decay rate $\theta$ on the lower-level game (UAVs' VSP selection distribution) and the upper-level game (VSPs' synchronization intensity) is shown in \cref{decay_all}. While keeping all the parameters the same as in the previous experiment, we set the learning rate $\delta$ as $0.05$ and vary the decay rate for VSP $1$, i.e., $\delta_1\in [0.05,0.2]$ in step size of $0.0167$. 
We plot populations states $\bm{x}$ in equilibrium and synchronization intensity $\bm{\eta}$ in equilibrium for all VSPs. We observe that the value of $\eta_1(t)$ increases as the decay rate of VSP $1$ increases, whereas $\eta_i(t),i=2,3,4$ decrease. The reason is that a higher decay rate of DTs indicates that VSP $1$ requires more synchronization data to maintain its DT status, and therefore VSP $1$ increases its synchronization intensity in response. Faced with a higher incentive offered by VSP $1$, more UAVs adjust their strategies to work for VSP $1$, thereby increasing $x_1(t)$.  However, subject to a limited number of UAVs assisting the Metaverse, data provisions to other VSPs are affected. Consequently, we observe a decrease in synchronization intensity for the remaining VSPs as well as fewer UAV selections. The result demonstrates the dynamic interactions between the two levels of the game, i.e., between the VSPs and UAVs. 
\begin{figure}[]
\centering
         \includegraphics[width=0.8\linewidth]{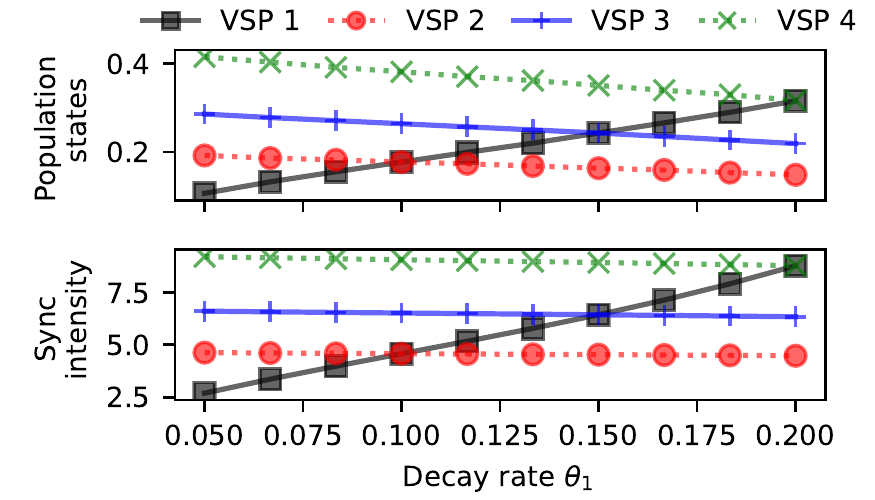}
\caption{Sensitivity analysis of DT value decay rate $\theta_1$}
\label{decay_all}
\end{figure}

\begin{figure}[]
\centering
\includegraphics[width=0.8\linewidth]{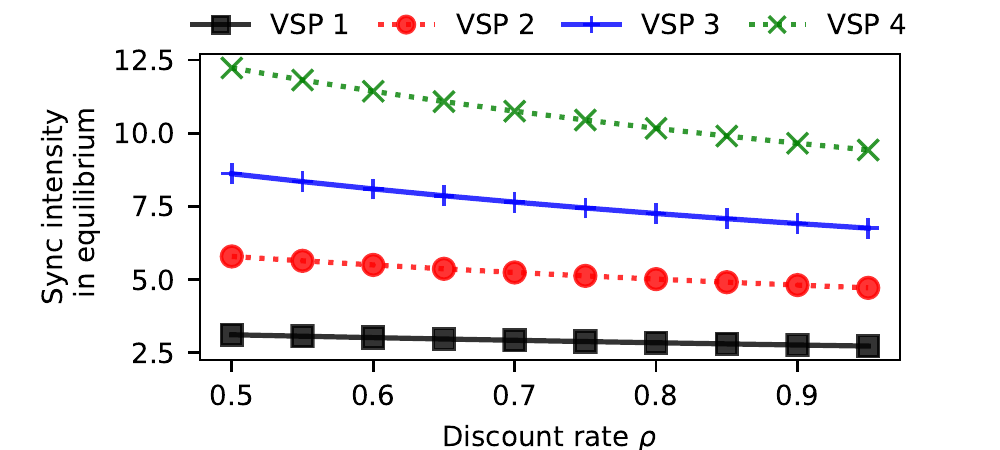}
\caption{Sensitivity analysis of time discount rate $\rho$}
\label{rho_eta}
\end{figure}

\begin{figure}[]
\centering
\includegraphics[width=0.8\linewidth]{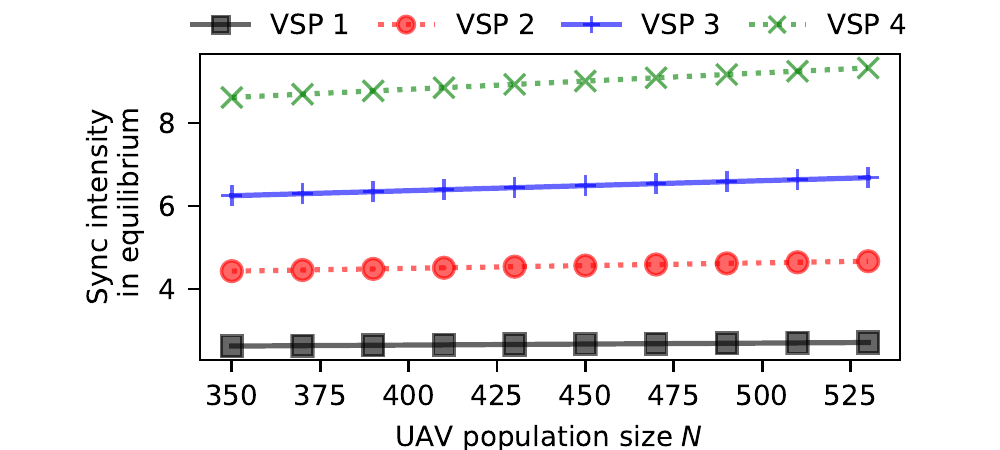}
\caption{ Sensitivity analysis of UAVs population size $N$}
\label{vary_n}
\end{figure}
\begin{figure}[]
\centering
\includegraphics[width=0.8\linewidth]{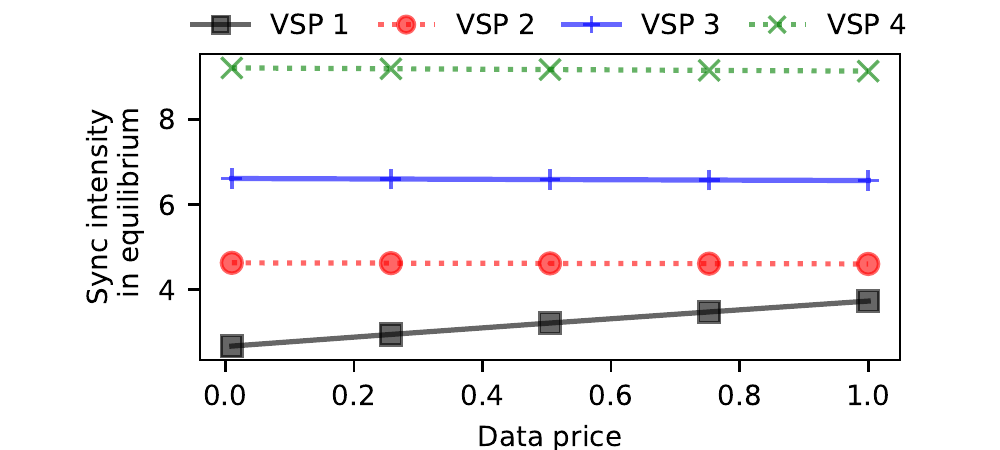}
\caption{ Sensitivity analysis of data price $\alpha_1$}
\label{data_price}
\end{figure}
\begin{figure}[]
\centering
\includegraphics[width=0.8\linewidth]{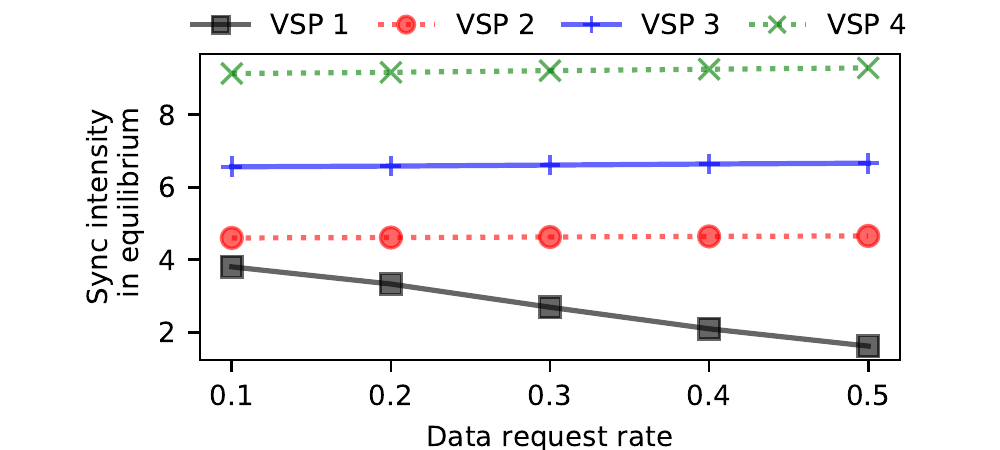}
\caption{ Sensitivity analysis of data request rate $k_1$}
\label{data_request_rate}
\end{figure}
\begin{figure}[]
\centering
\includegraphics[width=0.8\linewidth]{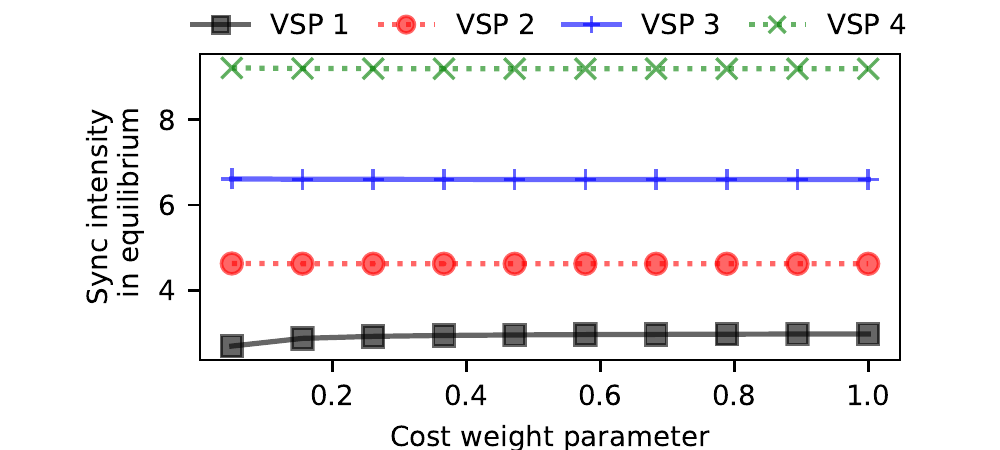}
\caption{ Sensitivity analysis of cost weight parameter $w^3_1$}
\label{zcost_eta}
\end{figure}

\begin{figure*}[t]
\centering
\includegraphics[width=0.65\linewidth]{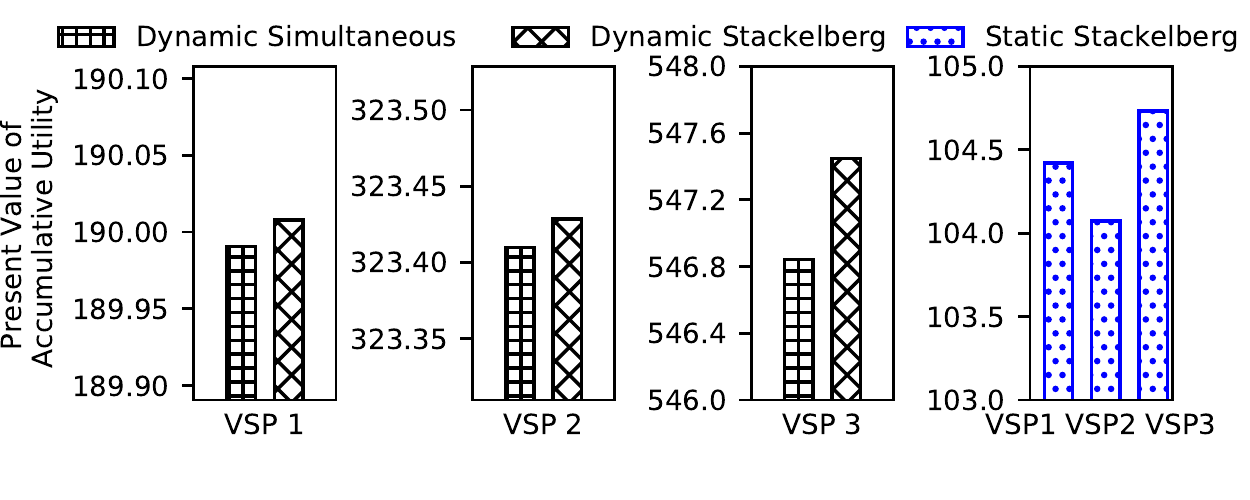}
\caption{ Comparison of the accumulative payoffs discounted at the present time for Stackelberg differential game, simultaneous differential game, and static Stackelberg game (baseline)}
\label{compareall}
\end{figure*}

\subsubsection{Varying $\rho$, the discount rate}
Figure \ref{rho_eta} shows the effect of VSPs' discount rate $\rho$ on the synchronization intensity $\bm{\eta}$ in equilibrium. We keep the experiment parameters the same as in the last experiment and set the value of $\delta_1$ to $0.05$. Then, we vary $\rho\in [0.5,1]$ in step size of $0.05$. We observe that the synchronization intensities for all VSPs decrease as the value of the discount factor increases. 
The reason is that the higher value of $\rho$, the lower the value of $e^{-\rho t}$ given a fixed $t$, thereby having a greater discounting effect on the future instantaneous utility at the time $t$. Therefore, $\mathfrak{J}_m,m=1,2,3,4$, the discounted cumulative payoffs, for all VSPs are smaller. In response, all VSPs lower their synchronization intensity.

\subsubsection{Varying $N$, the size of UAV population}
Figure \ref{vary_n} shows the effect of UAVs population size $n$ on the synchronization intensity $\bm{\eta}$ in equilibrium. We keep the experiment parameters the same as in the last experiment and set the value of $N$ from $350$ to $550$ in step size of $20$.
We observe that the synchronization intensities for all VSPs increase as $N$ increases. Because the higher value of $N$ is, the more data that can be supplied by the IoT group, thereby allowing the VSPs to have a higher synchronization intensity. Among the four VSPs, the increase for VSP $4$ is the most significant because its DTs have higher decay rates and therefore require higher synchronization data supply.

\subsubsection{Varying $\alpha_m$, the data price}

Figure \ref{data_price} shows the effect of data price $\alpha_m$ on the synchronization intensity $\bm{\eta}$ in equilibrium in the upper-level game. We keep the experiment parameters the same as in the last experiment and vary the value of $\alpha_1$ from $0.01$ to $1$. 
We observe that the synchronization intensities for VSP $1$ increase as the value of the $\alpha_1$ increases. Because the higher value of $\alpha_1$ indicates more accumulated payoffs can be obtained by VSP $1$, this allows VSP $1$ to afford a higher synchronization intensity. In contrast, the synchronization intensities for the remaining VSPs declined slightly. The reason is that more UAVs choose VSP $1$ and with a lower synchronized data supply, the synchronization intensities for the remaining VSPs drop.

\subsubsection{Varying $k_m$, the data request rate}
Figure \ref{data_request_rate} shows the effect of data request rate by DTs on the synchronization intensity $\bm{\eta}$ in equilibrium for the upper-level game. We keep the experiment parameter values the same as in the last experiment and vary the value of $k_1$ from $0.1$ to $0.5$. 
We observe that the synchronization rate for VSP $1$ decreases as the value of $k_1$ increases. The higher value of $k_1$ implies that more synchronization data are needed for one-time synchronization. In order to meet the demand for DTs based on the data supplied from UAVs, the VSPs choose to decrease the synchronization intensity such that demand and supply can be matched. 

\subsubsection{Varying $w^3_m$, the cost weight parameter in the objective functional}
Figure \ref{zcost_eta} shows the effect of the cost weight parameter $w^3_1$ on the synchronization intensity $\bm{\eta}$ in equilibrium in the upper-level game. $w^3_1$ is the weight parameter for the penalty term (i.e., cost to the VSPs) when the DTs' values are not meeting the threshold. We keep the experiment parameter values the same as in the last experiment and vary the cost weight parameter value of $w^3_1$ from $0.01$ to $0.03$. We observe that the synchronization rate for VSP $1$ increases as the value of  $w^3_1$ increases. The higher value of $w^3_1$ implies that more costs are incurred to VSP $1$ if the DTs' values are far away from the expected values. In response, VSP $1$ increases its synchronization intensity. 
Moreover, when the cost weight parameter is sufficiently large, the synchronization intensity at the equilibrium becomes unchanged.
The reason is that the cost incurred by low-valued DTs dominates the VSP $1$'s objective functional and therefore, VSP $1$'s equilibrium strategy becomes similar when the weight term is large.

\subsection{Comparison with the Hierarchical Play}
We conducted an experiment to compare three cases, \rom{1} simultaneous moves VSPs, i.e., a simultaneous differential game, \rom{2} hierarchical moves VSPs (i.e., the Stackelberg differential game, as presented in \cref{sec-upper-hierachical-play}), and \rom{3} a static Stackelberg game with an evolutionary game, as a benchmark. In the static Stackelberg game, there is no dynamic interaction between VSPs and UAVs: VSPs perform do a one-step optimization at the very beginning, and then UAVs populations evolve with the static synchronization intensity. In other words, in the static Stackelberg game, $\eta_m(t)=C_m,\forall t$, where $C_m$ is some constant that optimizes VSP $m$'s strategy. The steps to obtain a solution for a static Stackelberg game can be found in \cite{simaanStackelbergStrategyNonzerosum1973}. 
For illustration, we consider the case of one leader and two followers in the Stackelberg differential game and three simultaneous VSPs in the simultaneous differential game. 
As for the experiment setting, we consider VSPs $1$ to $3$ from the previous experiment and allow VSP $3$ to be the leader. Additionally, we set $\delta=0.02$ and $\rho=1$. 

As shown in \cref{compareall}, both the Stackelberg differential game and simultaneous differential game provide higher discounted cumulative payoffs than the static Stackelberg game. The reason is that both the Stackelberg differential game and simultaneous differential game capture the dynamic interactions between VSPs and UAVs and thus optimize the synchronization intensity control over time. It can also be observed that the Stackelberg differential game yields slightly higher discounted cumulative payoffs than the simultaneous differential game, especially for the leader. The reason is that the leader's decision can affect the followers' decisions, and thus the synchronization intensity control obtained by VSP $3$ (leader) under the Stackelberg differential game can help improve the overall discounted utilities. 

\section{Conclusion}
\label{sec-conclusion}
In this paper, we proposed a dynamic hierarchical framework to address the problem of DT synchronization for the virtual service providers in the Metaverse with assistance from UAVs. In particular, we propose a temporal value decay dynamics to measure the DT values and how they are affected by the VSPs synchronization strategy. In addition, a group of UAVs can assist VSPs to collect the most up-to-date status data of the physical counterparts to the DTs. Then, we adopted an evolutionary game to model the dynamic VSP selection behaviors for the population of UAVs. Open-loop Nash solutions were used to determine the optimal controls for the VSPs in the upper-level after formulating the upper-level problem as a differential game. To make the solution more realistic, we also considered the case where a group of VSPs can be the first movers in the market, and formulated it as a Stackelberg differential game. Experiments and proofs show that the equilibrium point of the lower-level game exists and is evolutionarily stable. In addition, the experiments demonstrate that dynamical games (both simultaneous differential game and Stackelberg differential game) outperform the static Stackelberg game. The equilibrium adaptation for different system parameters was also investigated. The extension to interoperability among VSPs will be considered in future work.

\section*{Acknowledgment}
This research was supported in part by the Alibaba Group through Alibaba Innovative Research (AIR) Program and Alibaba-Nanyang Technological University (NTU) Singapore Joint Research Institute(JRI); 
in part by the National Research Foundation, Singapore under its Emerging Areas Research Projects (EARP) Funding Initiative; 
in part by the National Research Foundation, Singapore, under AI Singapore Programme (AISG Award No: AISG-GC-2019-003);
in part by Singapore Ministry of Education (MOE) Tier 1 (RG16/20); 
and in part by National Research Foundation of Korea (NRF) Grant funded by the Korean Government (MSIT) under Grant 2021R1A2C2007638 and the MSIT under Grant IITP-2020-0-01821 supervised by the IITP.
\textit{}

\bibliographystyle{IEEEtran}
\bibliography{all2.bib}

\vspace{-1.5cm}

\newpage
%
%
%
%
%

\vfill

\end{document}